\documentclass[11pt]{article}

\usepackage[margin=0.9in]{geometry}
\usepackage{times}
\usepackage{amssymb, amsmath, amsthm, amsfonts}
\usepackage{enumitem}
\usepackage{nicefrac}
\usepackage{bbm}
\usepackage{bm}
\usepackage{xcolor}
\definecolor{myred}{HTML}{ea4335}
\definecolor{mygreen}{HTML}{41a756}
\definecolor{myblue}{HTML}{4285f4}
\usepackage{hyperref}
\hypersetup{
    colorlinks,
    allcolors  = mygreen
}
\usepackage[capitalize]{cleveref}

\usepackage{thm-restate}

\usepackage[sortcites,sorting=nyt,style=alphabetic]{biblatex}
\newtheorem{theorem}{Theorem}[section]
\newtheorem{corollary}[theorem]{Corollary}
\newtheorem{lemma}[theorem]{Lemma}
\newtheorem{proposition}[theorem]{Proposition}
\newtheorem{example}{Example}

\let\cd\cdot
\let\a\alpha
\let\g\gamma
\let\e\epsilon
\let\d\delta
\let\l\lambda
\let\sub\subseteq
\def\cblue{\color{myblue}}
\def\cred{\color{myred}}

\usepackage{xspace}
\def\nameS{Dynamic MMF\xspace}
\def\nameM{Dynamic DRF\xspace}
\def\nameW{Dynamic weighted MMF\xspace}

\newcommand{\N}{\mathbb{N}}

\newcommand{\ou}[3][=]{\overset{#2}{\underset{#3}{#1}}}
\newcommand{\ov}[2][=]{\overset{#2}{#1}}

\newcommand{\One}[1]{\mathbbm{1}\left[#1\right]}

\newcommand{\citewithauthor}[1]{\citeauthor{#1}~\cite{#1}}

\title{Incentives in Dominant Resource Fair Allocation under Dynamic Demands}

\author{
    Giannis Fikioris\thanks{Supported in part by NSF grants CCF-1408673 and AFOSR grant FA9550-19-1-0183.}\\
    Cornell University\\
    \texttt{gfikioris@cs.cornell.edu}
    \and
    Rachit Agarwal\thanks{Supported by NSF grants CNS-1704742 and CNS-2047283, and a Sloan fellowship.}\\
    Cornell University\\
    \texttt{ragarwal@cornell.edu}
    \and
    \'Eva Tardos\thanks{Supported in part by NSF grants CCF-1408673, CCF-1563714 and AFOSR grant FA9550-19-1-0183.}\\
    Cornell University\\
    \texttt{eva.tardos@cornell.edu}
}

\begin{document}

\maketitle
\begin{abstract}
    Every computer system---from schedulers in clouds (e.g., Amazon, Google, Microsoft, etc.) to computer networks to hypervisors to operating systems---performs resource allocation across system users. The defacto allocation policies used in most of these systems are max-min fairness for single resource settings and dominant resource fairness for multiple resources. These allocation schemes guarantee many desirable properties like incentive compatibility, envy-freeness, and Pareto efficiency, assuming user demands are static (time-independent). However, in modern real-world production systems, user demands are dynamic, that is, vary over time. As a result, there is now a fundamental mismatch between the resource allocation goals of computer systems, and the properties enabled by classical resource allocation policies. This paper aims to bridge this mismatch. When demands are dynamic, instant-by-instant max-min fairness can be extremely unfair over a longer period of time, i.e., lead to unbalanced user allocations, as previous large allocations have no effect in the current time step.
We consider a natural generalization of the classic algorithm for max-min fair allocation and dominant resource fairness for multiple resources when users have dynamic demands: this algorithm guarantees Pareto optimality while ensuring that resources allocated to users are as max-min fair as possible \emph{up to} any time instant, given the allocation in previous periods. While this dynamic allocation scheme remains Pareto optimal and envy free, unfortunately, it is not incentive compatible. We study the strength of the incentive to misreport; our results show that the possible increase in utility by misreporting demand is bounded and, since this misreporting can lead to significant decrease in overall useful allocation, this suggests that it is not a useful strategy. 
Our main result is to show that our dynamic version of the dominant resource fairness algorithm is $(1+\rho)$-incentive compatible, where $\rho$ quantifies the relative importance of a resource for different users; we also show that this factor is tight even with only two resources. 
We also present a $3/2$ upper bound and a $\sqrt 2$ lower bound for the incentive compatibility when there is only one resource. We also offer extensions of the results for the case  when the mechanism uses weights to prioritize every user differently. Our results indicate a big disparity between single resource and multiple resources even with just two resources showing a surprising difference in the users' incentive to deviate, with significantly less incentive to deviate in the case of one resource, and this is especially true in the case when the mechanism uses weights to prioritize every user differently.
\end{abstract}

\section{Introduction}
Resource allocation is a fundamental problem in computer systems. Companies like Google~\cite{borg} and Microsoft~\cite{DBLP:conf/sigcomm/GrandlAKRA14} use schedulers in private clouds to allocate a limited and divisible amount of resources (e.g., CPU, memory, servers, etc.) among a number of selfish and possibly strategic users that want to maximize their allocation; the goal of the scheduler is to maximize resource utilization while achieving fairness in resource allocation. The defacto allocation policies used in many of these systems are the classic \textit{max-min fairness} (MMF) and its relatively recent generalization, \textit{dominant resource fairness} (DRF)~\cite{DBLP:conf/nsdi/GhodsiZHKSS10} policies, for single and multiple resource settings, respectively. For instance, these policies are used in most schedulers in private clouds~\cite{yarn, mesos, borg, carbyene, apollo, DBLP:conf/nsdi/GhodsiZHKSS10, DBLP:conf/osdi/ShueFS12, DBLP:conf/sigcomm/GrandlAKRA14, DBLP:conf/osdi/GrandlKRAK16}; they are deeply entrenched in congestion control protocols like TCP and its variants~\cite{chiu1989analysis, dctcp}; and are the default policies for resource allocation in most operating systems and hypervisors~\cite{kvm, esxi}. DRF has also attracted a lot of attention in the economics and computing community, starting with \citewithauthor{DBLP:conf/sigecom/ParkesPS12} and with followup work \cite{DBLP:conf/ijcai/LiLL18, DBLP:conf/sigecom/FriedmanPV15, DBLP:conf/sigecom/FriedmanPV17, DBLP:conf/sigmetrics/ImMMP20, DBLP:conf/atal/KashPS13}.
%

Such prevalence of MMF and DRF is rooted in the properties they guarantee: Pareto-efficiency, sharing incentives (users are not better of by getting their fair share of resources every round), incentive compatibility, envy-freeness (no user envies the allocation of another user), and fairness. However, to guarantee these properties, both MMF and DRF policies assume that user demands do not change over time. This assumption is far from realistic in modern real-world deployments: several recent studies in the systems community have shown that user demands have become highly dynamic, that is, vary over time~\cite{vuppalapati2020building, cheng2018analyzing, reiss2012heterogeneity, yang2020large}. For such dynamic user demands, na\"ively using these policies (e.g., to perform a new instantaneously max-min fair allocation every unit of time) can result in vastly disparate user allocations over time---intuitively, since MMF does not take past allocations into account, dynamic user demands can result in increasingly unfair user allocations over time.

Motivated by the realistic case of dynamic user demands over divisible resources, we study {\nameM}{\footnote{The name Dynamic DRF has also been used by \cite{DBLP:conf/atal/KashPS13} for the extension of DRF when users arrive and depart sequentially. See more about the difference in \cref{sec:related}.}}, a mechanism that generalizes DRF for dynamic demands; just like DRF generalizes MMF for multi-resource allocations, Dynamic DRF generalizes {\nameS}~\cite{DBLP:journals/pomacs/FreemanZCL18} for multi-resource allocations over dynamic demands by taking past allocations into account. Our model is the same as the original DRF paper~\cite{DBLP:conf/nsdi/GhodsiZHKSS10}: every round, each user specifies a vector of \textit{ratios} (the proportions according to which the different resources are used by the user, e.g., for her application) and a \textit{demand} (the maximum allocation of resources that would be useful to her). Users have \textit{Leontief preferences}---as they are known in economics, i.e., their utility each round is equal to the minimum over the amount of every resource they receive divided by their ratio for it, up to their demand. However, in contrast to the DRF paper, we consider scenarios where the ratios and the demands can vary over time and users want to maximize the sum of their utilities across rounds.

In every round, \nameM allocates resources while being as fair as possible given the past allocations: first the minimum total utility of any user is maximized, then the second minimum, etc. Besides being fair, \nameM is also Pareto-efficient by construction: every round, either every user's demand is satisfied or for every user a resource she wants to use is saturated. However, similar to {\nameS}~\cite{DBLP:journals/pomacs/FreemanZCL18}, {\nameM} is not \textit{incentive compatible}, i.e., it is possible that a user can misreport her demand or her ratios on one round to increase her total useful allocation in the future (in fact, our study leads to stronger lower bounds on incentive compatibility of Dynamic MMF; see \cref{thm:single:lower,thm:mult:lower}). Nevertheless, studying {\nameM} is both important and interesting. First, similar to the widely-used classic MMF and DRF (also referred to as static MMF and DRF), {\nameM} is simple and easy to understand; thus, it has the potential for real-world adoption (similar to many other non-incentive compatible mechanisms used in practice, e.g., non-incentive compatible auctions used by U.S. Treasury to sell treasury bills since 1929 and by the U.K to sell electricity~\cite{krishna2009auction, harada2018, parkin2018}). Second, our results show that \nameM is approximately incentive compatible, that is, strategic users can increase their allocation by misreporting their demands but this increase is bounded by a relatively small constant factor, independent of the number of users and the number of resources; moreover, effective misreporting not only requires knowledge of future demands, but can lead to a significant decrease in overall useful allocation, suggesting that misreporting is unlikely to be a useful strategy for any user. 

\medskip
\noindent \textbf{Our Contribution.}
Our goal is to study \nameS and \nameM and the incentive to misreport in them. A popular relaxation of incentive compatibility is \textit{$\g$-incentive compatibility} \cite{DBLP:conf/soda/ArcherPTT03, DBLP:conf/sigecom/KotharPS03, DBLP:conf/soda/DekelFP08, DBLP:conf/sigecom/DuttingFJLLP12, DBLP:conf/sigecom/MennleS14, 10.1093/restud/rdy042, DBLP:conf/ec/BalcanSV19}, which requires that the possible increase in utility by untruthful reporting must be bounded by a factor of $\g \ge 1$ ($\g$ is referred to as \textit{incentive compatibility ratio}). Using this notion we show that users have limited incentive to be untruthful:
\begin{itemize}
    \item We start the technical part of this paper in \cref{sec:single} by considering the simpler problem of single resource environments, \nameS. We show that users have no incentive to over-report their demand (\cref{thm:single:no_over_report}), that \nameS is $3/2$-incentive compatible (\cref{thm:single:upper}), and give an almost matching lower bound of $\sqrt 2$ (\cref{thm:single:lower}).
    We also show that \nameS is envy-free (\cref{cor:mult:envy}) and the variant where every user is guaranteed an $\a$ fraction of her fair share of resources (for $\a\in[0,1]$) satisfies $\a$-sharing incentives while maintaining the incentive compatibility ratio, as well as every one of the aforementioned properties.
    
    \item Our main results are presented in \cref{sec:mult}, where we focus on the setting of multiple resources.
    In the case that every user demands every resource when using the system, we show that in \nameM user $i$ cannot increase her utility more than a factor of $(1+\rho_i)$ (\cref{thm:mult:upper}) and give a matching lower bound (\cref{thm:mult:lower}) where $\rho_i$ is a parameter that quantifies the relative importance of every resource between user $i$ and the other users. We also show that in this case users have no incentive to over-report their demand or misreport their ratios (\cref{thm:mult:no_over_report}); this guarantees that resources allocated to the users are always in use.
    The assumption that every user demands every resource when using the system, even if in different ratios, is quite natural; in computer systems where the resources shared are CPU, memory, storage, etc. the users run applications that use every resource. This assumption is also used by \cite{DBLP:conf/atal/KashPS13}, where they extend DRF to the case of users arriving and leaving sequentially resulting in dynamically changing the allocations in the system.
    
    Additionally, we show that \nameM is envy-free (\cref{thm:mult:envy}), and the variant where every user is guaranteed an $\a$ fraction of her fair share satisfies $\a$-sharing incentives while retaining every one of the aforementioned properties.
    
    \item We also consider some generalizations. First, we consider the case where every user $i$ is associated with a positive weight $w_i$ indicating their priority. This extra assumption does not change our results in \cref{sec:mult}, other than slightly altering the definition of $\rho_i$ affected by the weights $w$---the bound of $(1+\rho_i)$ remains tight for the new $\rho_i$. We also study weighted users in single resource settings in \cref{sec:coalitions}, along with the assumption that users can collude. In this case, \textit{\nameW} is $2$-incentive compatible (\cref{thm:gen:upper}) and demand over-reporting does not increase utility (\cref{thm:gen:no_over_report}). The former of these results strikes a big contrast between single and multiple resource settings: assuming that users do not collude and directly applying our results from \nameM in \nameW proves a $1 + \rho_i = 1 + \max_{k\ne i}(w_i/w_k)$ upper bound for the incentive compatibility of user $i$, which is possibly unbounded; in contrast, we manage to prove a $2$-upper bound.
    
\end{itemize}

\subsection{Related Work}\label{sec:related}

The simplest algorithm for resource allocation is {\em strict partitioning}~\cite{DBLP:conf/sigmod/VerbitskiGSBGMK17, DBLP:conf/nsdi/VuppalapatiMATM20}, that allocates a fixed amount of resources to each user independent of their demands. While incentive compatible, strict partitioning can have arbitrarily bad efficiency. Static MMF and DRF~\cite{DBLP:conf/nsdi/GhodsiZHKSS10, DBLP:conf/osdi/ShueFS12, DBLP:conf/sigcomm/GrandlAKRA14, DBLP:conf/osdi/GrandlKRAK16, DBLP:journals/pomacs/FreemanZCL18, DBLP:conf/sigecom/ParkesPS12} are Pareto-efficient, incentive compatible, envy-free, and satisfy sharing incentives, but fair only when user demands are static.

\citeauthor{DBLP:journals/pomacs/FreemanZCL18} \cite{DBLP:journals/pomacs/FreemanZCL18} prove that \nameS is not incentive compatible under the same utility model as ours.
The papers \cite{DBLP:journals/pomacs/FreemanZCL18, DBLP:conf/atal/Hossain19} study resource allocation for single resource settings with dynamic demands focusing on the case when users have utility for resources above their demand, only at a lower rate. They offer alternate mechanisms where past allocations have some effect on the current ones (unlike static MMF) while maintaining incentive compatibility, but the mechanisms they consider are closer to MMF separately in each epoch, and less aim to be fair overall. Under this model, they present two mechanisms that are incentive compatible but either satisfy sharing incentives and have no Pareto-efficiency guarantees, or approximately satisfy sharing incentives and are approximately Pareto-efficient under strong assumptions (user demands being i.i.d. random variables and number of rounds growing large).

\citewithauthor{DBLP:journals/tc/SadokCC21} present minor improvements in fairness over static DRF for dynamic demands while maintaining incentive compatibility. Their mechanism allocates resources in an incentive compatible way according to DRF while marginally penalizing users with larger past allocations using a parameter $\delta\in [0,1)$. Specifically, if $t$ rounds ago a user received an allocation of $r$, that allocation penalizes the user in the current round by $r(1-\d)\d^t$. This means that the penalty of $(1-\d)\d^t\le 0.25$ reduces exponentially fast with time for any fixed $\delta<1$ and as $\delta\to 1$ the penalty tends to $0$. Thus, for every $\delta$ (and, especially for $\delta=0$ and $\delta \to 1$), their mechanism suffers from similar problems as static DRF: past allocations are barely taken into account.

Several other papers study resource allocation where user demands can be dynamic, but with significantly different setting than ours. \cite{DBLP:conf/pricai/AleksandrovW19, DBLP:conf/sigecom/ZengP20} examine the setting where indivisible items arrive over time and have to be allocated to users whose utilities are random; however \citewithauthor{DBLP:conf/pricai/AleksandrovW19} study a very weak version of incentive compatibility in which a mechanism is incentive compatible if misreporting cannot increase a user's utility in the current round and \citewithauthor{DBLP:conf/sigecom/ZengP20} do not consider strategic users. \citewithauthor{DBLP:journals/corr/abs-2012-08648} study single resource allocation and assume the users do not know their exact demands every round and need to provide feedback to the mechanism after each round of allocation to allow the mechanism to learn.
The goal of the paper is to offer a version of MMF that approximately satisfies incentive compatibility, sharing incentives, and Pareto-efficiency, despite the lack of information, but is not considering the long-term fairness that is the focus of our mechanisms.

Another series of work study the setting where users arrive consecutively and potentially leave after some period of time: \cite{DBLP:conf/ijcai/LiLL18, DBLP:conf/sigecom/FriedmanPV15} focus on single resource settings, while \cite{DBLP:conf/sigecom/FriedmanPV17, DBLP:conf/sigmetrics/ImMMP20, DBLP:conf/atal/KashPS13} also study the allocation of multiple resources.
Even though their setting is dynamic, users have constant demands and cannot re-arrive after leaving, making the user demands static.
After every arrival or departure of a user, resources need to be re-distributed while maintaining some sort of fairness, e.g., the users' utilities need to be approximately similar. \citewithauthor{DBLP:conf/atal/KashPS13} focus on never decreasing a user's allocation when re-distributing resources, thus creating a mechanism that allocates at most $k/n$ fraction of every resource when $k$ out of $n$ users are present in the system. They offer a mechanism for this setting that they also call dynamic DRF, however, in their setting the dynamic nature of the problem comes from churn in users, as well as the corresponding changes in total resources, and not from dynamic demands.
The papers \cite{DBLP:conf/ijcai/LiLL18, DBLP:conf/sigecom/FriedmanPV15, DBLP:conf/sigecom/FriedmanPV17, DBLP:conf/sigmetrics/ImMMP20} consider resource allocation, but incentive compatibility is not taken into account and the focus is to maximize fairness while bounding the ``disruptions'' of the system every time a user enters or leaves, which is how many users' allocations are altered.

There are many reasons why mechanisms used in practice are often not incentive compatible, including the relative simplicity of these mechanisms that makes it easier for users to understand and use them and the fact that incentive compatible mechanisms are not actually incentive compatible depending on the information structure, e.g., when participants collude (see \citewithauthor{DBLP:conf/ec/BalcanSV19} for a nice discussion of a long list of other reasons). In most settings, even if the mechanism is manipulable, finding a profitable manipulation is hard. In our setting, finding such manipulation requires knowledge of all users' future demands, and while under-reporting demand can lead to increased future utility, it may also lead to decreased utility. When using a manipulable mechanism, it is important to understand how large is the incentive to manipulate.
$\g$-incentive compatibility has been considered in many settings.
As we mentioned before, \citewithauthor{DBLP:journals/corr/abs-2012-08648} study a setting similar to ours, but where long-term fairness is not their focus and they try to learn users' demands; to achieve their results they propose several mechanisms most of which are approximate incentive compatible.
\cite{DBLP:conf/soda/ArcherPTT03, DBLP:conf/sigecom/KotharPS03, DBLP:conf/sigecom/DuttingFJLLP12, DBLP:conf/sigecom/MennleS14} study combinatorial auctions that are almost incentive compatible. \citewithauthor{DBLP:conf/soda/DekelFP08} studies approximate incentive compatibility in machine learning, when users are asked to label data. \citewithauthor{10.1093/restud/rdy042} examines approximate incentive compatibility in large markets, where the number of users grows to infinity. \citewithauthor{DBLP:conf/ec/BalcanSV19} develops algorithms that can estimate how incentive compatible mechanisms for buying, matching, and voting are. \citewithauthor{DBLP:conf/ec/HartlineT19} develop auctions that use samples from past non-incentive compatible auctions to improve social welfare or revenue guarantees.
In the same spirit, \nameM is aimed at improving fairness.
\section{Preliminaries}\label{sec:prelims}

We first make some definitions that apply to all sections. We use $[n]$ to denote the set $\{1,2,\ldots,n\}$ for any $n\in\mathbb{N}$. Additionally, we define $x^+ = \max(x,0)$ and denote with $\One{\cd}$ the indicator function.

There are $n$ users, where $n\geq 2$. The set of users is denoted with $[n]$. In some settings, every user $i$ is associated with a weight, $w_i > 0$ which indicates each user's priority and in this case we view the allocation fair, if user $i$ has (approximately) $w_i/w_j$ more utility than user $j$. Additionally, we sometimes assume that user $i$ has initial endowment or \textit{fair share} a $w_i/\sum_{j\in[n]}w_j$ fraction of the total resources. 
The game is divided into \textit{epochs} $1, 2, \ldots, t, \ldots$

A mechanism is called \textit{envy-free} if for any pair of users $i$ and $j$, and if user $i$ is truthful, she would not have gained utility if she had been allocated the resources user $j$ was allocated.
Similarly, we define \textit{weighted envy-freeness}: if every user $i$ is associated with a weight $w_i>0$, a mechanism is envy-free if for any pair of users $i$ and $j$, user $i$ would not have gained utility if she had been allocated the resources user $j$ was allocated, scaled by $w_i/w_j$ (note that this scaling sometimes results in comparing to usage with more resources than what is available).

A mechanism satisfies \textit{sharing incentives} if every truthful user's utility is not less than her utility if she had been allocated her fair share every epoch, i.e., a $\frac{1}{n}$ fraction of the total resources.
For weighted users, a mechanism satisfies sharing incentives if the utility of user $i$ is not less than her utility if she had been allocated her fair share every epoch, i.e., a $\nicefrac{w_i}{\sum_{j\in[n]}w_j}$ fraction of the total resources.
For $\a\in[0,1]$, there is also the notion of \textit{$\a$-sharing incentives}, in which user $i$'s utility must be at least $\a$ times her utility if she had been allocated her fair share every epoch.
\section{Single Resource Setting}\label{sec:single}

Before we present our results on the multiple resource setting in \cref{sec:mult}, we first present results on the simpler setting of dynamic max-min fairness (\nameS) where there is only one resource and users are unweighted ($w_i = w_j$ for all users $i,j$).

\medskip
\noindent\textbf{Notation.}
Every epoch there is a fixed amount of a resource shared amongst the users. We denote the total amount of that resource with $\mathcal{R}$.

We denote with $r_i^t$ the \textit{allocation} of user $i$ in epoch $t$, i.e., the amount of resources the user receives. We also denote with $R_i^t$ the cumulative allocation of user $i$ up to round $t$, i.e., $R_i^t = \sum_{\tau=1}^t r_i^\tau$. By definition, $R_i^0=0$.

Every epoch $t$, each user $i$ has a \textit{demand}, denoted with $d_i^t$. This represents the maximum allocation that is useful for that user in that epoch, i.e., a user is indifferent between getting an allocation equal to her demand and an allocation higher than her demand, so the utility of user $i$ in epoch $t$ is $u_i^t = \min(r_i^t, d_i^t)$. The total utility of user $i$ after epoch $t$ equals the sum of utilities up to that round, i.e., $U_i^t = \sum_{\tau=1}^t u_i^\tau = \sum_{\tau=1}^t \min(r_i^\tau, d_i^\tau)$.

\paragraph{Dynamic Max-min Fairness}
In MMF the resources are allocated such that the minimum amount of resources is maximized, then the second minimum is maximized, etc., as long as every user gains an amount of resources that does not exceed her demand. If for example, we have $\mathcal{R}=1$ total resource and three users with demands $1/4$, $3/8$, and $1$, then the first user gets $1/4$ resources and the other two get $3/8$ resources each.

In \nameS, every epoch the MMF algorithm is applied to the users' cumulative allocations constrained by what they have already been allocated in previous iterations.
Our mechanism will have an additional parameter $\a\in[0,1]$ that guarantees a fraction of the fair share of each user every single epoch, independent of allocations in previous epochs: we guarantee every user at least $\a\mathcal R/n$ resources (assuming they have big enough demand to use it).
Formally, given an epoch $t$ and that every user $i$ has cumulative allocation $R_i^{t-1}$, \nameS solves the following problem
\begin{align}\label{eq:single:algo}
\begin{split}
    \textrm{choose }\;\; & r_1^t,\, r_2^t,\, \ldots,\, r_n^t \\
    \textrm{applying max-min fairness on }\;\; & R_1^{t-1} + r_1^t,\, R_2^{t-1} + r_2^t,\, \ldots,\, R_n^{t-1} + r_n^t \\
    \textrm{given the constraints }\;\; & 
        \sum_{i\in [n]} r_i^t \le \mathcal{R}\;\;\;\textrm{ and } \;\;\;
        \forall i\in [n],\, \min\left\{d_i^t, \a\frac{\mathcal R}{n}\right\} \le r_i^t \le d_i^t
\end{split}
\end{align}
We are going to call the value $\min\left\{d_i^t, \a\frac{\mathcal R}{n}\right\}$ the \textit{guarantee} of user $i$ in epoch $t$, and denote it with $g(d_i^t) = \min\left\{d_i^t, \a\frac{\mathcal R}{n}\right\}$. Note that $g(\cd)$ is a non-decreasing function. Because of the guarantee of every user it is easy to see that \nameS satisfies $\a$-sharing incentives.

\begin{theorem}\label{thm:single:si}
    When every user is guaranteed an $\a$-fraction of their fair share, \nameS satisfies $\a$-sharing incentives.
\end{theorem}
 
\paragraph{Incentives in \nameS}
It is easy to see \cite{DBLP:conf/sigecom/ParkesPS12} that applying MMF when there is a single epoch is \textit{incentive compatible}, i.e., users can never increase their utility by misreporting their demand. In dynamic settings, however, this is not the case. As was shown by \cite{DBLP:journals/pomacs/FreemanZCL18}, a user can increase her allocation by misreporting. See also the improved lower bound (\cref{thm:single:lower}).

\subsection{Bounds on Incentive Compatibility}

We will now focus in how far \nameS is from incentive compatible. W.l.o.g. we are usually  going to study the possible deviations of user $1$, i.e., how much user $1$ can increase her allocation by lying about her demand. We use the symbols $\hat d_i^t$, $\hat r_i^t$, $\hat R_i^t$, $\hat u_i^t$, $\hat U_i^t$ or $\bar d_i^t$, $\bar r_i^t$, $\bar R_i^t$, $\bar u_i^t$, $\bar U_i^t$ to denote the claimed demand and resulting outcome of some deviation of user $1$. We want to prove that for some $\g\ge 1$, \nameS is always \textit{$\g$-incentive compatible}, i.e., for every users' true demands $\{d_i^t\}_{i\in[n], t\in\N}$, for every deviation of user $1$ $\{\hat d_1^t\}_{t\in\N}$, and for every epoch $t$, to prove that $\hat U_1^t \le \g U_1^t$. $\g$ is often referred to as the \textit{incentive compatibility ratio}.

\paragraph{Lower bound on incentive compatibility}
As mentioned before, \nameS does not guarantee incentive compatibility. This is demonstrated in the following theorem, where user $1$ can misreport her demand to increase her utility by a factor of almost $\sqrt 2$.

\begin{restatable}{theorem}{counterexample}\label{thm:single:lower}
    For any value of $\a\in[0,1]$, there is an instance with $n$ users, where in \nameS a user can under-report her demand to increase her utility by a factor of $\sqrt 2$ as $n\to\infty$.
\end{restatable}

We defer the proof of the theorem to \cref{sec:app:single}. To provide intuition about how a user can increase her utility by misreporting, we include here the example of \cite{DBLP:journals/pomacs/FreemanZCL18}, where user $1$ can increase her utility by a factor of $10/9$.

\begin{example}[\cite{DBLP:journals/pomacs/FreemanZCL18}]\label{ex:counter_10_9}
    There are 3 users, 3 epochs, and $\a = 0$. Every epoch the available amount of the resource is $\mathcal R = 8$ and the real demands of the users are shown in \cref{tab:10_9_ex}, as well as their allocations when user $1$ is truthful and when she misreports.

    \begin{table}[h]
    \centering
    \begin{tabular}{c||ccc|ccc|ccc|}
    Users & \multicolumn{3}{c|}{Epoch $1$} & \multicolumn{3}{c|}{Epoch $2$} & \multicolumn{3}{c|}{Epoch $3$} \\\hline
    User $1$ &
        $8$ & \cblue $4$ & \cred $0$ &
        $8$ & \cblue $2$ & \cred $4$ &
        $8$ & \cblue $3$ & \cred $6$ \\
    User 2 &
        $8$ & \cblue $4$ & \cred $8$ &
        $0$ & \cblue $0$ & \cred $0$ &
        $8$ & \cblue $5$ & \cred $2$ \\
    User 3 &
        $0$ & \cblue $0$ & \cred$0$ &
        $8$ & \cblue $6$ & \cred $4$ &
        $0$ & \cblue $0$ & \cred$0$
    \end{tabular}
    \caption{Every epoch there is a total of $8$ resources. The black numbers are the users' demands, the blue ones are the users' allocations when user $1$ is truthful, and the red ones are the allocations when user $1$ misreports her demand in epoch $1$ by demanding $0$.}
    \label{tab:10_9_ex}
    \end{table}
    
    Because user $1$ under-reports her demand in epoch $1$, in epoch $2$ she manages to ``steal'' some of user $3$'s resources. Then, in epoch $3$ the allocation mechanism equalizes the total allocations of users $1$ and $2$, making user $1$ get back some of the resources she lost in epoch $1$. This results in user $1$ having a total allocation of $10$ instead of $9$, i.e., her utility increases by a factor of $10/9$.
    \qed
\end{example}

Both in \cref{thm:single:lower} and \cref{ex:counter_10_9}, it is important to note that user $1$ can increase her utility only by a bounded constant factor. Additionally, this is done by user $1$ under-reporting her demand, not over-reporting; this is important because it implies that the resources allocated are always used by the users and hence the allocation remains Pareto optimal. As we will show next, both of these facts are true in general.

\paragraph{Upper bound on incentive compatibility ratio}
To prove the above, first we show a lemma offering a simple condition on which pair of users can gain overall allocations from one another. When users' demands are not satisfied, for a user to get more resources someone else needs to get less. The lemma will allow us to reason about how a deviation by user $1$ can lead to a user $i$ getting more resources and another user $j$ getting less (possibly $i=1$ or $j=1$).

\begin{lemma}\label{lem:single:more_less}
    Fix an epoch $t$ and the total allocations up to epoch $t-1$ of any two outcomes $\{\hat R_k^{t-1}\}_{k\in [n]}$ and $\{\bar R_k^{t-1}\}_{k\in [n]}$. Let $i,j$ be two different users. If the following conditions hold
    \begin{itemize}
        \item For $i$,
        $\,\bar r_i^t < \hat r_i^t$ and $\hat d_i^t \le \bar d_i^t$.
        \item For $j$,
        $\,\bar r_j^t > \hat r_j^t$ and $\bar d_j^t \le \hat d_j^t$.
    \end{itemize}
    then, for any value of $\a\in [0,1]$, in \nameS it holds that $\bar R_i^t \geq \bar R_j^t$ and $\hat R_i^t \le\hat R_j^t$, implying
    \begin{equation*}
        \hat R_i^t - \bar R_i^t \le \hat R_j^t - \bar R_j^t
    \end{equation*}
\end{lemma}

It should be noted that the second condition for $i$ (similarly for $j$) is not needed when $i$ is not the user who misreports her demand, i.e., $i\neq 1$.

\begin{proof}
    Because of the conditions, we notice that $\bar r_i^t < \bar d_i^t$ (since $\hat r_i^t\le \hat d_i^t$) and $\bar r_j^t > g(\bar d_j^t)$ (since $\bar r_j^t > \hat r_j^t \ge g(\hat d_j^t) \ge g(\bar d_j^t)$). These two inequalities imply that it would have been feasible in \eqref{eq:single:algo} to increase $\bar r_i^t$ by decreasing $\bar r_j^t$. This implies that $\bar R_i^t \geq \bar R_j^t$; otherwise it would have been more fair to give some of the resources user $j$ got to user $i$. 
    With the analogous inverse argument (we can increase $\hat r_j^t$ by decreasing $\hat r_i^t$) we can prove that $\hat R_i^t \le \hat R_j^t$. This completes the proof.
\end{proof}

The main technical tool in our work is the following lemma bounding the total amount all the users can ``win'' because of user $1$ deviating, i.e., $\sum_k(\hat R_k^t - R_k^t)^+$. Rather than bounding the deviating user $1$'s gain directly, it is better to consider the overall increase of all users.
Specifically, the lemma upper bounds the increase of that amount after any epoch, given that user $1$ does not over-report her demand (which as we are going to show later in \cref{thm:single:no_over_report} users have no benefit in doing). The bound on the total over-allocation then follows by summing over the epochs. The bound on the increase of $\sum_k(\hat R_k^t - R_k^t)^+$ after any epoch is different according to two different cases: 
\begin{itemize}
    \item If user $1$ is truthful, the increase is at most $0$: in these steps over-allocation can move between users but cannot increase. This is the reason studying the total over-allocation is so helpful.
    \item If user $1$ under-reports, the increase is bounded by the amount of resources she receives when she is truthful.
\end{itemize}

\begin{lemma}\label{lem:single:aux_bound}
    Fix an epoch $t$ and the total allocations of any two outcomes $\{\hat R_k^{t-1}\}_{k\in [n]}$ and $\{\bar R_k^{t-1}\}_{k\in [n]}$. Assume that $\{\bar d_k^t\}_{k\in [n]}$ are some user demands and that $\{\hat d_k^t\}_{k\in [n]}$ are the same demands except user~$1$'s, who deviates but not by over-reporting, i.e., $\hat d_1^t \le \bar d_1^t$. Then, for any $\a\in[0,1]$, it holds that
    \begin{equation}\label{eq:single:aux}
        \sum_{k\in [n]}\left(\hat R_k^t - \bar R_k^t\right)^+
        -
        \sum_{k\in [n]}\left(\hat R_k^{t-1} - \bar R_k^{t-1}\right)^+
        \le
        \bar r_1^t \One{\hat d_1^t < \bar d_1^t}
    \end{equation}
\end{lemma}

We will use \cref{lem:single:more_less} to show that if user $1$ is truthful in epoch $t$, then the l.h.s. of \eqref{eq:single:aux} is at most $0$: \nameS allocates resources such that the large $\hat R_k^t - \bar R_k^t$ are decreased and the small $\hat R_k^t - \bar R_k^t$ are increased. Finally, if user $1$ reports a lower demand, then the (at most) $\bar r_1^t$ resources user $1$ does not get might increase the total over-allocation by the same amount. 

\begin{proof}
    First we define $P^t=\{k\in [n]: \hat R_k^t > \bar R_k^t\}$ for all $t$. Suppose by contradiction:
    \begin{equation*}
        \sum_{k\in P^t}\left(\hat R_k^t - \bar R_k^t\right)
        -
        \sum_{k\in P^{t-1}}\left(\hat R_k^{t-1} - \bar R_k^{t-1}\right)
        >
        \bar r_1^t \One{\hat d_1^t < \bar d_1^t}
    \end{equation*}
    
    Because $\sum_{k\in P^{t}}(\hat R_k^{t-1} - \bar R_k^{t-1})\le\sum_{k\in P^{t-1}}(\hat R_k^{t-1} - \bar R_k^{t-1})$, the above inequality implies
    \begin{equation}\label{eq:single:11}
        \sum_{k\in P^t}\left(\hat r_k^t - \bar r_k^t\right) > \bar r_1^t \One{\hat d_1^t < \bar d_1^t}
    \end{equation}

    Since user $1$ reports a lower demand $\hat d_1^t \le \bar d_1^t$, it holds that $\sum_k \bar r_i^t \ge \sum_k \hat r_i^t$, i.e., the total resources allocated to the users do not increase when user $1$ reports a lower demand. Combining this fact with \eqref{eq:single:11} we get that
    \begin{equation}\label{eq:single:12}
        \sum_{k\notin P^t}\left(\bar r_k^t - \hat r_k^t\right) > \bar r_1^t \One{\hat d_1^t < \bar d_1^t}
    \end{equation}

    We notice that because of \eqref{eq:single:11}, there exists a user $i\in P^t$ for whom $\hat r_i^t > \bar r_i^t$; because of \eqref{eq:single:12}, there exists a user $j\notin P^t$ for whom $\bar r_j^t > \hat r_j^t$. Additionally for that $j$ we can assume that $\hat d_j^t = \bar d_j^t$ because:
    \begin{itemize}
        \item If user $1$ does not deviate then for all $k$, $\hat d_k^t = \bar d_k^t$.
        
        \item If $\hat d_1^t < \bar d_1^t$, then \eqref{eq:single:12} becomes $\sum_{k\notin P^t,\,k\neq 1}(\bar r_k^t - \hat r_k^t) > 0$, implying $j\neq 1$ which entails $\hat d_j^t = \bar d_j^t$.
    \end{itemize}
    
    Thus we have $\hat d_i^t \le \bar d_i^t$ (since no user over-reports), $\hat d_j^t = \bar d_j^t$, $\hat r_i^t > \bar r_i^t$, and $\hat r_j^t < \bar r_j^t$. These, using \cref{lem:single:more_less}, prove that $\hat R_i^t - \bar R_i^t \le \hat R_j^t - \bar R_j^t$. This is a contradiction because $i\in P^t$ and $j\notin P^t$.
\end{proof}

Before proving the upper bound on the incentive compatibility ratio, first we will prove that users have no incentive to over-report their demand. The immediate effect of over-reporting is allocating resources to user $1$ in excess of her demand, which do not contribute to her utility. Intuitively, this suggests that user $1$ is put into a disadvantageous position: other users get less resources which makes them be favored by the allocation algorithm in the future, while user $1$ becomes less favored. However, a small change in the users' resources causes a cascading change in future allocations making the proof of this theorem hard. We will see in \cref{sec:mult} that this is no longer the case with multiple resources when users only use a subset of them. We defer the proof of the theorem to the end of this section.

\begin{theorem}\label{thm:single:no_over_report}
    Users have nothing to gain in \nameS by declaring a demand higher than their actual demand, for any value $\a\in[0,1]$.
\end{theorem}

Now we show that using this theorem and \cref{lem:single:aux_bound} allows us to bound the incentive to deviate. A bound of $2$ is easy to get by summing \eqref{eq:single:aux} for all $t$ up to some certain epoch and assuming that $\bar d_i^t$ are users' true demands ($\bar d = d$) and that $\hat d_1$ is any under-reporting of user $1$. We give an incentive compatibility bound of $3/2$ by using the same lemma, but arguing that some other user $j$ must also share the same increased allocation of resources as user $1$ using \cref{lem:single:more_less}.

\begin{theorem}\label{thm:single:upper}
    In \nameS for any value of $\a\in [0,1]$, no user can misreport her demand to increase her utility by a factor larger than $3/2$, i.e., for any user $i$ and demand misreporting user $i$ makes, $\hat U_i^t \le \frac 3 2 U_i^t$, for all epochs $t$.
\end{theorem}

\begin{proof}
    We will prove the theorem for $i=1$. \cref{thm:single:no_over_report} implies that it is of no loss of generality to assume that user $1$ does not over-report her demand, since any benefit gained by over-reporting can be gained by changing every over-report to a truthful one. This means that instead of $\hat U_1^t \le \frac 3 2 U_1^t$ we can show $\hat R_1^t \le \frac 3 2 R_1^t$. Towards a contradiction, let $t$ be the first epoch when user $1$ gets more than $3/2$ more resources by some deviation of demands, i.e., 
   $\hat R_1^t > \frac 3 2 R_1^t$ and $\hat R_1^{t-1} \le \frac 3 2 R_1^{t-1}$. This implies that $\hat r_1^t > r_1^t$, which in turn entails that there exists a user $j$ for who $\hat r_j^t < r_j^t$, since the total resources allocated when user $1$ is under-reporting cannot be less than those when 1 is truthful. Because $\hat r_1^t > r_1^t$, $\hat r_j^t < r_j^t$, $\hat d_1^t \le d_1^t$, and $\hat d_j^t = d_j^t$, we can use \cref{lem:single:more_less} and get $\hat R_1^t - R_1^t \le \hat R_j^t - R_j^t$. This inequality, $\hat R_1^t - R_1^t \geq 0$, and \cref{lem:single:aux_bound} by summing \eqref{eq:single:aux} for every epoch up to $t$, implies
    \begin{equation*}
        2\left(\hat R_1^t -  R_1^t\right)
        \;\le\;
        \left(\hat R_1^t -  R_1^t\right) + \left(\hat R_j^t -  R_j^t\right)
        \;\le\;
        \sum_{k\in[n]}\left(\hat R_k^t -  R_k^t\right)^+
        \;\le\;
        \sum_{\tau=1}^t r_1^\tau
        \;=\;
        R_1^t
    \end{equation*}

    The above inequality leads to $\hat R_1^t \le \frac 3 2 R_1^t$, a contradiction.
\end{proof}

Finally, we prove \cref{thm:single:no_over_report}, that users have no incentive to over-report. To do that, we use the following lemma, which says that if a user wants to increase her utility in epoch $T$, she has nothing to gain by over-reporting her demand in epoch $T_0$, given that she does not over-report her demand in epochs $T_0+1$ to $T$. 

\begin{lemma}\label{lem:single:over_aux}
    Fix an epoch $T_0$ and the allocations of an outcome $\{\hat R_k^{T_0-1}\}_{k\in [n]}$. Fix another epoch $T\ge T_0$ and assume that in epochs $t = T_0+1,T_0+2,\ldots,T$ user $1$ does not over-report her demand, i.e.,~$\hat d_1^t \le d_i^t$. Then user $1$ cannot increase her utility in round $T$ by over-reporting her demand in epoch $T_0$, for any $\a\in[0,1]$.
\end{lemma}

By over-reporting her demand in $T_0$, user $1$ (potentially) gains some resources that do not contribute to her utility, while others get less resources. This puts user $1$ in a disadvantage entailing that in the epochs after $T_0$ she cannot increase her total allocation further (even though it is possible that her allocation in a single epoch can increase). Since she doesn't over-report her demand in epochs after $T_0$ her allocation is the same as her utility, hence her total utility also does not increase.

\begin{proof}
    To prove the lemma we are going to create another outcome in which user $1$ changes her over-report in epoch $T_0$ to a truthful one which does not decrease her total utility in epoch $T$.
    
    For all users $k$ and epochs $t$, let $\bar d_k^t = \hat d_k^t$, except for $\bar d_1^{T_0}$, which is user $1$'s actual demand: $\bar d_1^{T_0} = d_1^{T_0}$. This means that the two outcomes are the same before epoch $T_0$, i.e., for all users $k$, $\bar R_k^{T_0-1} = \hat R_k^{T_0-1}$ and $\bar U_k^{T_0-1} = \hat U_k^{T_0-1}$. Because $\hat d_1^{T_0} > \bar d_1^{T_0}$ and for $k\neq 1$, $\hat d_k^{T_0} = \bar d_k^{T_0}$, user $1$ may earn some additional resources on $T_0$, i.e., $\hat r_1^{T_0} - \bar r_i^{T_0} = \hat R_1^{T_0} - \bar R_i^{T_0} = x$, for some $x\ge 0$, while other users $k\neq 1$ get less resources: $\hat r_k^{T_0} - \bar r_k^{T_0} = \hat R_k^{T_0} - \bar R_k^{T_0} \le 0$. We first note that the $x$ additional resources that user $1$ gets are in excess of $1$'s true demand, meaning they do not contribute towards her utility:
    \begin{equation}\label{eq:single:1}
        \hat U_1^{T_0} - \bar U_1^{T_0} = \hat R_1^{T_0} - \bar R_1^{T_0} - x = 0
    \end{equation}
    
    Additionally, because user $1$ does not over-report $\hat d_1^t$ or $\bar d_1^t$ in epochs $t\in[T_0+1,T]$, it holds that for $t\in [T_0+1, T]$ user $1$'s utility is the same as the resources she receives: $\bar u_1^t = \bar r_1^t$ and $\hat u_1^t = \hat r_1^t$.  This fact, combined with \eqref{eq:single:1} proves that
    \begin{equation*}
        \forall t\in[T_0,T],\;
        \hat U_1^t - \bar U_1^t = \hat R_1^t - \bar R_1^t - x
    \end{equation*}
    
    Thus, in order for this over-reporting to be a strictly better strategy, it must hold that $\hat R_1^T - \bar R_1^T > x$. We will complete the proof by proving that the opposite holds. Since in epochs $t\in[T_0+1, T]$ it holds that $\hat d_1^t = \bar d_1^t$, we can use \cref{lem:single:aux_bound} to sum \eqref{eq:single:aux} for all $t\in [T_0+1, T]$ and get that
    \begin{equation*}
        \sum_k\left(\hat R_k^T - \bar R_k^T\right)^+
        -
        \sum_k\left(\hat R_k^{T_0} - \bar R_k^{T_0}\right)^+
        \le 0
    \end{equation*}
    
    The above inequality, because $(\hat R_k^T - \bar R_k^T)^+\ge 0$, $\hat R_k^{T_0} - \bar R_k^{T_0} \le 0$ for $k\neq 1$, and $\hat R_1^{T_0} - \bar R_1^{T_0} = x \ge 0$, proves that $\hat R_1^T - \bar R_1^T \le x$. This completes the proof.
\end{proof}


\begin{proof}[Proof of \cref{thm:single:no_over_report}]
    Fix an epoch $T$ and let $T_0\le T$ be the last epoch where user $1$ over-reported. \cref{lem:single:over_aux} allows us to change user $1$'s over-report in $T_0$ to a truthful one without decreasing her total utility in $T$. Doing this inductively for every such $T_0$ creates a demand profile with no over-reporting that does not decrease user $1$'s total utility in $T$.
\end{proof}

\noindent
\textbf{Envy-freeness in \nameS.}
We finally note that \nameS is envy-free, for any value of the parameter $\a$. This is a corollary of the fact that \nameM, the generalization of \nameS for multiple resources, is envy-free (see details in \cref{thm:mult:envy,cor:mult:envy}).
\section{Multiple Resources Setting}\label{sec:mult}

In this section we turn to the main focus of our paper, analyzing the generalization of \nameS for multiple resources, Dynamic Dominant Resource Fairness (\nameM). 

\medskip
\noindent\textbf{Notation and User Utilities.}
We consider users that have varying demand for a set of $m\ge 1$ different resources over time. We use $1,2,\ldots,q,\ldots,m$ to denote the $m$ resources, and w.l.o.g., we assume that for every resource the amount available is the same, $\mathcal{R}$. A typical example of such a system may focus on users running applications that use resources such as CPU, memory, storage, etc.

Every epoch, each user demands a non-negative amount of every resource which they report to the mechanism. With the multidimensional nature of demand, users have more complex ways to misreport their demand. Throughout this paper we will assume that users have \textit{Leontief preferences}, which we define next. Leontief preferences have been considered by much of the previous work in resource sharing, including \cite{DBLP:conf/nsdi/GhodsiZHKSS10, DBLP:conf/sigecom/ParkesPS12, DBLP:conf/atal/KashPS13}.

Formally, with Leontief preferences a user $i$'s demand in an epoch $t$ is characterized by the \textit{ratios} $a_i^t = (a_{i1}^t,\ldots,a_{im}^t)$ she needs for the resources and a \textit{demand} $d_i^t$. The ratios indicate that for some $\xi\ge 0$, user $i$'s application in that epoch is going to use $\xi a_{iq}^t$ amount of every resource $q\in [m]$. The resource $q$ with the maximum ratio $a_{iq}^t$ is called the \textit{dominant resource} of user $i$ in epoch $t$. We assume that ratios are normalized: for every epoch $t$, $\max_q a^t_{iq}=1$ for all users. The demand $d_i^t$ of user $i$ in epoch $t$ represents the maximum fraction $\xi$ (possibly, $\xi>1$) of the ratios the user can take advantage of. More specifically, user $i$ demands (or is asking for) $d_i^t a_{iq}^t$ amount of every resource $q$.
If a user $i$ receives $x_{i1}^t,\ldots,x_{im}^t$ of every resource, respectively, her utility that epoch is $u_i^t = \min\big\{d_i^t, \min_{q: a_{iq}>0} \{x^t_{iq}/a^t_{iq}\}\big\}$.
In each epoch, users will be asked to report both their ratios for the epoch as well as their demand. Note that users can now misreport their type in two different ways: they can either request less or more from all the resources or they can demand the resources in different proportions (or they can do both). 

\paragraph{Dynamic Dominant Resource Fairness}
DRF is the generalization of MMF for the case of multiple resources, where the fairness criterion is applied to each user's dominant resource and the rest of the resources are distributed according to the users' ratios. If $r_i^t$ is the amount that user $i$ receives of her dominant resource in epoch $t$, then she receives $r_i^t a_{iq}^t$ of every resource $q$ (recall that $\max_q a_{iq}^t = 1$). We call $r_i^t$ the \textit{allocation} of user $i$ in epoch $t$.

\nameM extends to dynamic user demands the weighted DRF policy, the generalization of DRF where the fairness criterion is applied to the users' allocation normalized by some weights, indicating users' priorities. More specifically, if every user $i$ is associated with a weight $w_i>0$, then the mechanism tries to give to each user $i$ an allocation proportional to her weight $w_i$.

Similar to \nameS, \nameM has an additional parameter $\a\in[0,1]$ that guarantees a fraction of fair share of each resource to each user in every epoch, independent of allocations in previous ones. User $i$'s fair share of a resource $q$ is $\mathcal R w_i/(\sum_j w_j)$. When $\alpha=1$, we guarantee at least the fair share of their dominant resource (assuming a big enough demand to use it), when $\alpha<1$ we guarantee a smaller share. Beyond this guarantee, the goal of the mechanism in epoch $t$ is to be as fair as possible to the cumulative allocation of every user, normalized by their weights. We use $R_i^t =\sum_{\tau=1}^t r_i^\tau$ for the sum of allocations till time $t$. Using this notation, \nameM is easy to describe. For a given epoch $t$, assuming that every user $i$ has cumulative allocation $R_i^{t-1}$:
\begin{align}\label{eq:mult:algo}
\begin{split}
    \textrm{choose }\;\;
        & r_1^t,\, r_2^t,\, \ldots,\, r_n^t \\
    \textrm{applying max-min fairness on }\;\;
        & \frac{R_1^{t-1} + r_1^t}{w_1},\, \frac{R_2^{t-1} + r_2^t}{w_2},\, \ldots,\, \frac{R_n^{t-1} + r_n^t}{w_n} \\
    \textrm{given the constraints }\;\;
        & \forall i\in [n],\, \min\left\{d_i^t, \a \frac{\mathcal R w_i}{\sum_{k\in[n]} w_k}\right\}\leq r_i^t \leq d_i^t,\\
        & \forall q\in [m], \sum_{i\in [n]} r_i^t\, a_{iq}^t \leq \mathcal{R}
\end{split}
\end{align}
We define $g_i(d_i^t) = \min\left\{d_i^t, \a \frac{\mathcal R w_i}{\sum_k w_k}\right\}$ to be the guaranteed amount that every user receives every epoch. Note that $g_i(\cd)$ is a non-decreasing function.

We note a few properties that immediately follow from the description. If all users share the same dominant resource $q^t$ in each epoch and have equal weights, then the mechanism will become identical to \nameS, as at each epoch the allocation of the shared dominant resource is the bottleneck for all users.

Second, if all the users share their dominant resource $q^t$, $\alpha=1$, and demands are high, then the minimum guarantee becomes $g(d_i^t) = \mathcal R w_i / \sum_k w_k$ which will become user $i$'s allocation, as this saturates resource $q^t$ (the guaranteed total use of $q^t$ is $\sum_i g(d_i^t)a_{iq^t}^t =\sum_i g(d_i^t) = \mathcal R$). However, even with large demands each iteration and $\alpha=1$, the dynamic fair sharing nature of our allocation will play an important role when applications (users) do not always share their dominant resource.

Third, because of the guarantee of every user, \nameM satisfies $\a$-sharing incentives.

\begin{theorem}\label{thm:mult:si}
    When every user is guaranteed an $\a$-fraction of their fair share, \nameM satisfies $\a$-sharing incentives.
\end{theorem}


\paragraph{User's utility when misreporting ratios}
In defining the \nameM mechanism, we have not considered the difference of truthful reporting and misreporting one's demands or ratios to the mechanism. The main topic of this section is understanding how a user's utility behaves in these two scenarios for \nameM.

When user $i$ truthfully reports her demand $d_i^t$ and ratios $a_i^t$, and \nameM gives her an allocation of $r_i^t$, her utility in in that epoch is $u_i^t = \min\{d_i^t, \min_{q:a_{iq}^t > 0} \nicefrac{r_i^t a_{iq}^t}{a_{iq}^t}\} = r_i^t$, since \nameM guarantees that $r_i^t \le d_i^t$.

When user $i$ misreports $\hat a_i^t$ and $\hat d_i^t$, and \nameM gives her an allocation $\hat r_i^t$ based on the reported values, let $\hat u_i^t$ denote the user's true utility in epoch $t$ under that reporting. In this case she receives $x_{iq}^t = \hat a_{iq}^t \hat r_i^t$ of every resource $q$ and thus she gets true utility 
\begin{equation*}
    \hat u_i^t = \min \left\{d_i^t,\min_{q: a_{iq^t}>0} \frac{x^t_{iq}}{a^t_{iq}}\right\}
    = \min\left\{
        d_i^t,\,
        \hat r_i^t\min_{q: a_{iq^t}>0} \frac{\hat a_{iq}^t}{a_{iq}^t}
    \right\}
\end{equation*}
We define $\hat\l_i^t = \min_{q: a_{iq^t}>0} \{\hat a_{iq}^t/a_{iq}^t\}$ making the above expression equal to $\hat u_i^t = \min\{d_i^t,\,\hat r_i^t\,\hat\l_i^t\}$. 
We should note that if the user reports ratios truthfully ($\hat a_i^t = a_i^t$), then $\hat\l_i^t = 1$. Additionally, because each user is constrained to declare $\max_q \hat a_{iq}^t = \max_q a_{iq}^t = 1$, it must hold that $\hat\l_i^t \le 1$ for any $\hat a_i^t$ .

Similar to \cref{sec:single}, the total utility of user $i$ by epoch $t$ is $U_i^t = \sum_{\tau=1}^t u_i^\tau$.

\subsection{Incentives assuming positive ratios for all resources}
\label{sec:mult:nonzero}

Our main results in this paper consider \nameM resource allocation with multiple resources under the assumption that when users want to use the system, they demand at least some of each of the resources, that is $a_{iq}^t>0$ for all $i,q,t$. With typical system resources, such as CPU, memory, RAM, etc., it is indeed the likely scenario. While different applications have different dominant resources (e.g., some have heavy use of compute power, while in others the main bottleneck is bandwidth), each uses at least some of each one of these basic resources.

As already observed by \citewithauthor{DBLP:conf/sigecom/ParkesPS12}, having zero ratios for some resources significantly changes the problem from having a tiny $\epsilon>0$ ratio, no matter how small $\epsilon$ is. In Appendix \ref{sec:app:mult}we show that with zero ratios users do have an incentive to over-report their demand, which will turn out to not be useful with positive ratios. Further, the benefit of such over-reporting can increase the user's utility by a factor of $\Omega(m)$, increasing with the number of resources in the system.

In contrast, the main results of this section are that, assuming users have positive ratios, misreporting them, as well as over-reporting demand, is not beneficial (\cref{thm:mult:no_over_report}). Further, the approximate incentive compatibility ratio for user $1$ is bounded by $1+\rho_1$, where $\rho_1 = \max_{k\ne 1,q,t}\big\{\nicefrac{w_1 a_{1q}^t}{w_k a_{kq}^t}\big\}$ (\cref{thm:mult:upper}), and this bound is tight (\cref{thm:mult:lower}) even with just two resources and static ratios, extending the results of \cref{sec:single} to multiple resources.

\paragraph{Upper bound on incentive compatibility ratio}
We start by pointing out why positive ratios are so different from zero ones. With positive ratios, users are bottlenecked by the same resource being saturated. In contrast, when ratios are zero, different users are bottlenecked by different resources, resulting in users' allocations being almost independent from one another (which is the case if users do not share resources for which they both have positive ratios). The difficulty in adapting the results of Section \ref{sec:single} to the case of multiple resources comes from the fact that different users use the resources in different proportions. For example, consider the users' allocations being bottlenecked by a single saturated resource $q$ for which user $1$ has a ratio of $1$ and user $i$ has a ratio of $\e < 1$. User $1$ can free a $\d$ amount of resource $q$ by under-reporting her demand so that her allocation gets smaller by $\d$, for a small $\d>0$. If this available amount of resources goes to user $i$, it has the potential to raise her allocation by $\d/\e$ (recall that user $i$ has a ratio of $\e$ for resource $q$), which is disproportionally bigger than the $\d$ utility that user $1$ lost.

The assumption that all users are using each resource allows us to prove a lemma analogous to \cref{lem:single:more_less}, as with this assumption, we can increase any user's allocation by decreasing another user's allocation, allowing the proof to be identical to the proof of \cref{lem:single:more_less}. In contrast, if $a_{iq}^t = 0$ and $a_{jq}^t > 0$, then decreasing the allocation of user $i$ does not free any amount of resource $q$ for user $j$. Proving this lemma will lead to results similar to \cref{thm:single:no_over_report,thm:single:upper}.

\begin{restatable}{lemma}{multMoreLess}\label{lem:mult:more_less}
    Fix an epoch $t$ and the total allocations up to epoch $t-1$ of any two outcomes $\{\hat R_k^{t-1}\}_{k\in [n]}$ and $\{\bar R_k^{t-1}\}_{k\in [n]}$. Let $i,j$ be two different users. If the following conditions hold
    \begin{itemize}
        \item For $i$,
        $\,\bar r_i^t < \hat r_i^t$, $\,\hat d_i^t \le \bar d_i^t$, and $\hat a_{iq}^t > 0$ for all $q$.
        \item For $j$,
        $\,\bar r_j^t > \hat r_j^t$, $\,\bar d_j^t \le \hat d_j^t$, and $\bar a_{jq}^t > 0$ for all $q$.
    \end{itemize}
    then, for any $\a\in[0,1]$ used by \nameM, it holds that $\bar R_i^t/w_i \geq \bar R_j^t/w_j$ and $\hat R_i^t/w_i \le\hat R_j^t/w_j$, implying
    \begin{equation*}
        \frac{\hat R_i^t - \bar R_i^t}{w_i} \le \frac{\hat R_j^t - \bar R_j^t}{w_j}
    \end{equation*}
\end{restatable}

The proof of this lemma is identical to the one in \cref{lem:single:more_less}, so it's deferred in \cref{sec:app:mult}.

Now we present an auxiliary lemma, similar to \cref{lem:single:aux_bound}, but this time we bound the increase of any user's allocation when user $1$ deviates, i.e., $\hat R_k^t - R_k^t$ for each $k$. This lemma shows the key difficulty in extending the results to the multiple resource case: reporting low demand in one epoch can result in a different user being able to get much higher amounts, even when the two users have the same dominant resource.
Unfortunately, as a result of this difficulty the resulting bound is a weaker version of \cref{lem:single:aux_bound}, involving the parameter $\rho_1$. If the users ratios are the same and users have same priority, i.e., $w_i = w_j$ and $a_{i}^t = a_{j}^t$ for every $t$ and $i\neq j$, then $\rho_1=1$ and the following lemma would allow us to prove a $2$-incentive compatibility ratio upper bound. In general however, the incentive compatibility ratio depends on $\rho_1$ and it becomes larger the larger $\rho_1$ is.

\begin{restatable}{lemma}{multAux}\label{lem:mult:aux_bound}
    Fix an epoch $t$ and the allocations of two different outcomes $\{\hat R_k^{t-1}\}_{k\in [n]}$ and $\{\bar R_k^{t-1}\}_{k\in [n]}$ for which it holds that for some $X\ge 0$,
    \begin{equation}\label{eq:mult:hypothesis}
        \forall k\in[n],\;\; \frac{\hat R_k^{t-1} - \bar R_k^{t-1}}{w_k} \le X 
    \end{equation}
    If in epoch $t$ users have positive ratios, user $1$ reports her ratios truthfully, i.e.,~$\hat a_1^t = \bar a_1^t = a_1^t$, and $\hat d_1^t \le \bar d_1^t$ then, for all $\a\in[0,1]$ used by \nameM it holds that
    \begin{equation*}
        \forall k\in[n],\;\; 
        \frac{\hat R_k^t - \bar R_k^t}{w_k}
        \leq
        X + 
        \One{\hat d_1^t < \bar d_1^t} 
        \rho_1^t \frac{\bar r_1^t}{w_1}
    \end{equation*}
    where $\rho_1^t = \max_{ k\ne 1,\,q\in [m] } \frac{w_1 a_{1q}^t}{w_k a_{kq}^t}$.
\end{restatable}

The above lemma, by assuming truthful reporting of ratios and that user $1$ does not over-report her demand ($\hat d_1^t \le d_1^t$), inductively proves a $\big(1+\max_t\{\rho_1^t\}\big)$ incentive compatibility ratio.

The proof of this lemma is similar to the one in \cref{lem:single:aux_bound}. If $\hat d_1^t = \bar d_1^t$, then users with larger $\nicefrac{\hat R_k^{t-1} - \bar R_k^{t-1}}{w_k}$ will not gain additional resources. If $\hat d_1^t < \bar d_1^t$, user $1$ can increase the allocation of some other user $k$, but by at most $\rho_1^t \bar r_1^t \frac{w_k}{w_1}$; note that this is much larger than $r_1^t$, the guarantee of \cref{lem:single:more_less}.

\begin{proof}
    Assume that there exists a user $i$ such that
    \begin{equation}\label{eq:mult:41}
        \frac{\hat R_i^t - \bar R_i^t}{w_i}
        >
        X + 
        \One{\hat d_1^t < \bar d_1^t} 
        \rho_1^t \frac{\bar r_1^t}{w_1}
    \end{equation}

    Subtracting \eqref{eq:mult:hypothesis} from \eqref{eq:mult:41} implies that
    \begin{equation}\label{eq:mult:42}
        \frac{\hat r_i^t - \bar r_i^t}{w_i}
        >
        \One{\hat d_1^t < \bar d_1^t} 
        \rho_1^t \frac{\bar r_1^t}{w_1}
    \end{equation}

    Note that \eqref{eq:mult:42} implies that user $i$ satisfies the requirements of user $i$ in \cref{lem:mult:more_less} (since for all users $k$ and resources $q$, $\hat d_k^t \le \bar d_k^t$ and $a_{kq}^t>0$).
    Now note that because for all users $k$, $\hat d_k^t \le \bar d_k^t$, there must exist a resource $q^*$ for which
    \begin{equation*}
        \sum_k a_{kq^*}^t\, \bar r_k^t
        \ge
        \sum_k a_{kq^*}^t\, \hat r_k^t
    \end{equation*}
    since otherwise no resources would be saturated in allocations $\{\bar r_k^t\}_k$, which would imply that every user gets her demand in that allocation, which would imply the same in allocations $\{\hat r_k^t\}_k$ (since ratios are the same and demands are not higher); the last fact leads to a contradiction.
    Multiplying \eqref{eq:mult:42} with $a_{iq^*}^t w_i$ and adding it with the above inequality yields
    \begin{equation}\label{eq:mult:43}
        \sum_{k\ne i}a_{kq^*}^t (\bar r_k^t - \hat r_k^t)
        \;>\;
        w_i a_{iq^*}^t \One{\hat d_1^t < \bar d_1^t} 
        \rho_1^t \frac{\bar r_1^t}{w_1}
        \;\ge\;
        0
    \end{equation}
    \eqref{eq:mult:43} proves that there exists a user $j$ such that $\bar r_j^t > \hat r_j^t$. Additionally, for that $j$ we can assume that $\bar d_j^t = \hat d_j^t$; to see why that is we distinguish three cases:
    \begin{itemize}
        \item If $\hat d_1^t = \bar d_1^t$ then for all users $k$, $\hat d_k^t = \bar d_k^t$.
        
        \item If $i = 1$, then $j\ne 1$ since $\hat r_i^t > \bar r_i^t$ and $\bar r_j^t > \hat r_j^t$. Because $j\ne 1$, $\hat d_j^t = \bar d_j^t$.
        
        \item If $\hat d_1^t < \bar d_1^t$ and $i\ne 1$, then, \eqref{eq:mult:43} becomes
        \begin{equation*}
            \sum_{k\ne i}a_{kq^*}^t (\bar r_k^t - \hat r_k^t)
            \;>\;
            w_i a_{iq^*}^t \rho_1^t \frac{\bar r_1^t}{w_1}
            \;\ou[\ge]{i\ne 1\implies}{\rho_1^t \ge \nicefrac{w_1 a_{1q^*}^t}{w_i a_{iq^*}^t}}\;
            a_{1q^*}^t \bar r_1^t
        \end{equation*}
        which implies that $\sum_{k\ne i,1}a_{kq^*}^t (\bar r_k^t - \hat r_k^t) > 0$. This proves that we can assume that $j\ne 1$, entailing $\hat d_j^t = \bar d_j^t$.
    \end{itemize}

    \noindent
    This proves that exists a user $j$ that satisfies the requirements of user $j$ from \cref{lem:mult:more_less}. Using the lemma on users $i$ and $j$ we get $\frac{\hat R_i^t - \bar R_i^t}{w_i} \le \frac{\hat R_j^t - \bar R_j^t}{w_j}$.
    This is a contradiction, because it holds that $\hat R_j^t - \bar R_j^t < \hat R_j^{t-1} - \bar R_j^{t-1} \le w_j X$ (from \eqref{eq:mult:hypothesis} and $\bar r_j^t > \hat r_j^t$)
    and $\hat R_i^t - \bar R_i^t > w_i X$ (from \eqref{eq:mult:41}).
\end{proof}

Similarly to \cref{thm:single:no_over_report}, we can now prove that there is no benefit to over-reporting or misreporting ratios. As mentioned previously, this is a very important property because it guarantees that every resource allocated is utilized and that, even when user $1$ is not truthful, her utility is equal to her allocation: $\hat u_i^t = \hat r_i^t$.

\begin{theorem}\label{thm:mult:no_over_report}
    Assume that the users' true ratios are positive, i.e., for all users $i\in [n]$, resources $q\in [m]$, and epochs $t$, it holds that $a_{iq}^t>0$. Then, for any $\a\in[0,1]$ used by \nameM, the users have nothing to gain by declaring a demand higher than their actual demand, and any gain achievable by misreporting ratios can also be obtained by under-reporting demand.
\end{theorem}

The proof of this theorem is quite similar to the one in \cref{thm:single:no_over_report}: first we prove the following auxiliary lemma, that makes the proof of the theorem easy. The lemma is similar to \cref{lem:single:over_aux}: if in epochs $T_0+1$ to $T$ user $1$ does not over-report her demand and reports her ratios truthfully, she cannot increase her utility in epoch $T$ by over-reporting demand or misreporting ratios in $T_0$. 

\begin{lemma}\label{lem:mult:over_aux}
    Fix an epoch $T_0$ and the allocations of an outcome $\{\hat R_k^{T_0-1}\}_{k\in [n]}$. Fix another epoch $T\ge T_0$ and assume that in epochs $T_0+1,T_0+2,\ldots,T$ user $1$ reports her ratios truthfully and does not over-report her demand. Then, for any $\a\in[0,1]$ used by \nameM, any increase in user $1$'s utility in epoch $T$ gained by over-reporting demand or misreporting ratios in $T_0$ can be achieved with truthful ratio reporting and no demand over-reporting in epoch $T_0$.
\end{lemma}

The above proves that over-reporting demand or misreporting ratios is not useful by applying the lemma inductively, similar to the proof of \cref{thm:single:no_over_report}.

The lemma's proof is based on an alternative reporting that both reports ratios truthfully and guarantees no demand over-reporting. With this reporting user $1$ is guaranteed the same utility without increasing her allocation or decreasing the other user's allocations; this entails a more advantageous position for her in the following epochs.

\begin{proof}[Proof of \cref{lem:mult:over_aux}]
    We are going to consider another outcome that is the same as the one in the lemma up to epoch $T_0-1$. All users $k\ne 1$ report the same type in all epochs: $\bar d_k^t = \hat d_k^t$ and $\bar a_k^t = \hat a_k^t$ for all $t$. In epoch $T_0$ user $1$ deviates from reporting $\hat a_1^{T_0}$ and $\hat d_1^{T_0}$, and instead declares
    \begin{itemize}
        \item her ratios truthfully: $\bar a_1^{T_0} = a_1^{T_0}$.
        \item her demand: $\bar d_1^{T_0} = \min\{d_1^{T_0}, \hat\l_1^{T_0}\hat r_1^{T_0}\}$; recall that $\hat\l_1^t = \min_q \frac{\hat a_{1q}^t}{a_{1q}^t} \le 1$.
    \end{itemize}

    Note that the above reporting satisfies what we want: it reports the ratios truthfully and does not over-report the demand.
    In later epochs $t\in [T_0+1,T]$ user $1$ does not deviate from the other reporting: $\bar a_1^t = \hat a_1^t = a_1^t$ and $\bar d_1^t = \hat d_1^t \le d_1^t$.

    Now we notice that in epoch $T_0$, for every resource $q$ it holds that
    \begin{equation}\label{eq:mult:44}
        \bar a_{1q}^{T_0}\, \bar r_1^{T_0}
        \;\ou[\le]{\bar a_1^{T_0} = a_1^{T_0}}{\bar r_1^{T_0} \le \bar d_1^{\,T_0}}\;
        a_{1q}^{T_0}\, \bar d_1^{T_0}
        \;\ou[\le]{\textrm{definition}}{\textrm{of }\bar d_1^{T_0}}\;
        a_{1q}^{T_0}\, \hat\l_1^{T_0}\, \hat r_1^{T_0}
        \;\ou[\le]{\textrm{definition}}{\textrm{of }\hat \l_1^{T_0}}\;
        a_{1q}^{T_0}\, \frac{\hat a_{1q}^{T_0}}{a_{1q}^{T_0}}\, \hat r_1^{T_0}
        \;=\;
        \hat a_{1q}^{T_0}\,  \hat r_1^{T_0}
    \end{equation}

    \eqref{eq:mult:44} proves that in allocation $\{\bar r_k^{T_0}\}_k$ user $1$ uses less of every resource compared to allocation $\{\hat r_k^{T_0}\}_k$. This proves that the following allocation is feasible: $\bar r_1^{T_0} \gets \bar d_1^{T_0}$ and for $k\ne 1$, $\bar r_1^{T_0} \gets \hat r_1^{T_0}$. Because $\bar d_1^{T_0} \le \hat r_1^{T_0}$ (which follows from $\hat\l_1^{T_0}\le 1$) we have that
    \begin{equation}\label{eq:mult:45}
        \bar r_1^{T_0} = \bar d_1^{\,T_0}
        \;\;\textrm{ and }\;\;
        \forall k\ne 1,\;
        \bar r_k^{T_0}
        \ge
        \hat r_k^{T_0}
    \end{equation}
    Now we notice that 
    \begin{equation*}
        \hat u_i^{T_0}
        \;\ou{\textrm{definition}}{\textrm{of utility}}\;
        \min\left\{
            d_i^{T_0},\,
            \hat r_i^{T_0}\,\hat\l_1^{T_0}
        \right\}
        \;\ou{\textrm{definition}}{\textrm{of }\bar d_1^{T_0}}\;
        \bar d_1^{\,T_0}
        \;\ov{\eqref{eq:mult:45}}\;
        \bar r_1^{T_0}
        \;\ou{\bar d_1^{T_0}\le d_1^{T_0}}{\bar a_1^{T_0} = a_1^{T_0}}\;
        \bar u_1^{T_0}
    \end{equation*}
    The difference in total utility in epoch $T$ is
    \begin{eqnarray*}
        \hat U_1^T - \bar U_1^T
        &=&
        \hat U_1^{T_0-1} - \bar U_1^{T_0-1} +
        \hat u_1^{T_0} - \bar u_1^{T_0} +
        \sum_{\tau=T_0+1}^T\left(\hat u_1^\tau - \bar u_1^\tau\right) \\
        &\ou{\textrm{same outcomes up to } T_0-1}{\hat u_1^{T_0} = \bar u_1^{T_0}}&
        \sum_{\tau=T_0+1}^T\left(\hat u_1^\tau - \bar u_1^\tau\right) \\
        &\ou{\textrm{no demand over-reporting}}{\textrm{and true ratio reporting in } \tau}&
        \sum_{\tau=T_0+1}^T\left(\hat r_1^\tau - \bar r_1^\tau\right)
        \;\ou{\hat R_1^{T_0-1} = \bar R_1^{T_0-1}}{}\;
        \hat R_1^T - \bar R_1^T
        -
        \left(\hat r_1^{T_0} - \bar r_1^{T_0}\right)
    \end{eqnarray*}%

    Let $X = \hat r_1^{T_0} - \bar r_1^{T_0} \ge 0$. To prove that $\hat U_1^T \le \bar U_1^T$ (thus proving the lemma) all we need to prove is that $\hat R_1^T - \bar R_1^T \le X$. For every epoch $t \in [T_0+1,T]$ the requirements of \cref{lem:mult:aux_bound} are true in $t$ and $\hat d_1^t = \bar d_1^t$, which proves inductively that if
    \begin{equation}\label{eq:mult:47}
        \forall k\in [n] : \frac{\hat R_k^{T_0} - \bar R_k^{T_0}}{w_k} \le \frac{X}{w_1} 
    \end{equation}
    then for all $t\in [T_0,T]$
    \begin{equation*}
        \forall k\in [n] : \frac{\hat R_k^t - \bar R_k^t}{w_k} \le \frac{X}{w_1}
    \end{equation*}

    We can prove \eqref{eq:mult:47} by noticing that for $k \ne 1$, it holds $\hat R_k^t - \bar R_k^t \le 0 \le X$ (from \eqref{eq:mult:45} and $\hat R_k^{T_0-1} = \bar R_k^{T_0-1}$) and for $k=1$ it holds that $\hat r_1^{T_0} - \bar r_1^{T_0} = X$ and $\hat R_1^{T_0-1} = \bar R_1^{T_0-1}$.
    This completes the proof.
\end{proof}

Now we prove the upper bound for the incentive compatibility of \nameM with a simple proof utilizing \cref{lem:mult:aux_bound,thm:mult:no_over_report}.

\begin{theorem}\label{thm:mult:upper}
    Assume that all users have positive ratios: $a_{iq}^t>0$ for all users $i$, resources $q$, and epochs $t$.  Then for any user $i$ and $\a\in[0,1]$ used by \nameM, user $i$ cannot misreport her demand or ratios to increase her utility by a factor larger than $(1 + \rho_i)$, where $\rho_i = \max_{k\ne i,q,t}\Big\{\frac{w_i a_{iq}^t}{w_k a_{kq}^t}\Big\}$.
\end{theorem}

\begin{proof}
    W.l.o.g. we are going to prove the theorem for $i=1$. Because of \cref{thm:mult:no_over_report} we can assume that user $1$ does not over-report her demand or misreport her ratios and thus we can bound $\hat R_1^t$ instead of $\hat U_1^t$. \cref{lem:mult:aux_bound} inductively implies that for any $t$, $\hat R_1^t - R_1^t \le \rho_1 R_1^t$, which proves the theorem.
\end{proof}

\paragraph{Lower bound on incentive compatibility ratio}
In the last theorem we proved that the incentive compatibility ratio of user $1$ is upper bounded by $(1+\rho_1)$. We now show that if the only constraints on users' ratios and weights is that they are positive and $\rho_1$ is fixed, then it is possible for the incentive compatibility ratio of user $1$ to be $(1+\rho_1)$, even if the users' ratios do not change over time.

\begin{restatable}{theorem}{multLower}\label{thm:mult:lower}
    For any $\e\in (0,1)$, $w_1,w_2>0$, and $\a\in[0,1]$ used by \nameM, there is an instance where the users' ratios are constant every epoch, user $1$ has weight $w_1$ and another user has weight $w_2$, $\rho_1 = \frac{w_1}{w_2 \e}$, and user $1$ can under-report her demand to increase her utility by a factor of $1+\rho_1$.
\end{restatable}

We briefly sketch the proof of the theorem when users have time-varying ratios, there is only one resource, users have equal weights, and it does not necessarily hold that $\max_{q\in[m]}a_{iq}^t = 1$. With some complicated details the full proof reduces to the same instance as the one we present below.

\begin{proof}[Proof sketch]
    There are two users, one resource, two epochs, $w_1=w_2=1$, and $\a = 0$.
    \begin{enumerate}
        \item In epoch $1$, user $1$ has a ratio of $a_1^1 = 1$, user $2$ has a ratio of $a_2^1 = \e$, the amount available of the resource is $1$, and both users have infinite demand. \nameM allocates the resource such that $r_1^1 = r_2^1 = \frac{1}{1+\e}$.
        
        \item In epoch $2$, user $1$ has a ratio of $a_1^2 = \d \le 1$, user $2$ has a ratio of $a_2^2 = 1$, the amount available of the resource\footnote{We can assume that in every epoch there is a different total amount of resources without loss of generality: we can split an epoch into $x$ sub-epochs where in each there are $1/x$ available resources.} is $\d/\e$, and both have infinite demand. \nameM, since the users' past allocations are the same, allocates the resource such that $r_1^2 = r_2^2 = \frac{\d/\e}{1+\d}$.
    \end{enumerate}
    
    This results in $R_1^2 = R_2^2 = \frac{1}{1+\e} + \frac{\d/\e}{1+\d}$. We note that $
        \rho_1
        =
        \max\left\{ \frac{a_1^1}{a_2^1}, \frac{a_1^2}{a_2^2} \right\}
        =
        \max\left\{ \frac{1}{\e}, \frac{\d}{1} \right\}
        =
        \frac{1}{\e}$.
    
    Now we study what happens when user $1$ deviates:
    \begin{enumerate}
        \item If in epoch $1$ user $1$ demands $0$ resources, then \nameM allocates $\hat r_1^1 = 0$ and $\hat r_2^1 = 1/\e$.
        \item In epoch $2$, the previous misreporting results in user $1$ having a lower total allocation, making \nameM allocate $\hat r_1^2 = 1/\e$ and $\hat r_2^2 = 0$.
    \end{enumerate}
    
    This results in $\hat R_1^2 = \hat R_2^2 = 1/\e$. The incentive compatibility ratio is
    \begin{equation}\label{eq:mult:lower_ratio}
        \frac{\hat R_1^2}{R_1^2}
        =
        \frac{1/\e}{\frac{1}{1+\e} + \frac{\d/\e}{1+\d}}
        \ou[\to]{\d\to 0}{}
        1 + \frac{1}{\e}
        =
        1 + \rho_1
    \end{equation}
\end{proof}

We defer the full proof in \cref{sec:app:mult}.
We should note that in the above proof sketch and in the full proof of \cref{thm:mult:lower}, even though $\rho_1 = 1/\epsilon$, a ratio of user $1$ is $\d$, which is much less than $\e$. If we wanted a bound in terms of $\rho = \max_{i\in [n]} \rho_i$, the proof of \cref{thm:mult:lower} for $w_1 = w_2$ (and the proof sketch above, see \eqref{eq:mult:lower_ratio} by taking $\d=\e$) easily yields a $\frac{1+\rho}{2}$ lower bound.


\subsection{Envy-freeness}

In this section we prove that \nameM is envy-free.

\begin{theorem}\label{thm:mult:envy}
    For every $\a\in[0,1]$, \nameM is envy-free according to the weights $w_1,\ldots,w_n$, i.e., for every epoch $t$, no user $i$ envies the total allocation of user $j$ scaled by $w_i/w_j$:
    \begin{equation*}
        U_i^t
        \ge
        \sum_{\tau=1}^t \min\left\{ d_i^\tau,\, \frac{w_i}{w_j} r_j^\tau 
            \min_{q:a_{iq}^\tau>0} \frac{a_{jq}^\tau}{a_{iq}^\tau}
        \right\}
    \end{equation*}
\end{theorem}

The key to this proof is to study the (potential) first epoch $t$ user $i$ envies user $j$. In the simple case that $w_i = w_j$, that proves that $r_j^t > r_i^t$, which leads to a reasoning similar to \cref{lem:mult:more_less} proving that $R_j^t \le R_i^t$. However, since user $j$ has an allocation smaller than $i$'s it is impossible for user $i$ to envy her.

\begin{proof}
    Towards a contradiction assume that in epoch $t$ user $i$ envies user $j$ for the first time, i.e.,
    \begin{equation}\label{eq:mult:75}
        U_i^t
        =
        R_i^t
        <
        \sum_{\tau=1}^t \min\left\{ d_i^\tau,\, \frac{w_i}{w_j} r_j^\tau 
            \min_{q:a_{iq}^\tau>0} \frac{a_{jq}^\tau}{a_{iq}^\tau}
        \right\}
    \end{equation}
    and
    \begin{equation*}
        U_i^{t-1}
        =
        R_i^{t-1}
        \ge
        \sum_{\tau=1}^{t-1} \min\left\{ d_i^\tau,\, \frac{w_i}{w_j} r_j^\tau 
            \min_{q:a_{iq}^\tau>0} \frac{a_{jq}^\tau}{a_{iq}^\tau}
        \right\}
    \end{equation*}%
    Subtracting the two inequalities above we get
    \begin{equation*}
        r_i^t
        <
        \min\left\{ d_i^t,\, \frac{w_i}{w_j} r_j^t
            \min_{q:a_{iq}^t > 0} \frac{a_{jq}^t}{a_{iq}^t}
        \right\}
    \end{equation*}
    Using the above inequality we can prove that
    \begin{itemize}
        \item $\min_{q: a_{iq}^t > 0} \nicefrac{a_{jq}^t}{a_{iq}^t} > 0$, i.e., user $j$ uses every resource that user $i$ uses.
        \item $r_i^t < d_i^t$. This inequality proves that $g_i(d_i^t) = \a\mathcal R w_i /(\sum_k w_k)$, since otherwise $g_i(d_i^t) = d_i^t = r_i^t$.
        \item $r_i^t/w_i < r_j^t/w_j$, since $\max_q a_{iq}^t = \max_q a_{jq}^t = 1$. This inequality, since $r_i^t \ge g_i(d_i^t) = \a\mathcal R w_i /(\sum_k w_k)$, proves that $r_j^t > \a\mathcal R w_j /(\sum_k w_k) \ge g_j(d_j^t)$.
    \end{itemize}
    These three facts, $r_i^t < d_i^t$, $r_j^t > g_j(d_j^t)$, and $\forall q: a_{iq}^t > 0 \implies a_{jq}^t > 0$, prove that it would have been feasible to increase $r_i^t$ (thus increasing $R_i^t$) by decreasing $r_j^t$ (thus decreasing $R_j^t$).
    This proves that, because our mechanism allocates resources max-min fairly, $R_i^t/w_i \ge R_j^t/w_j$. The last inequality contradicts with \eqref{eq:mult:75}, which proves that $R_i^t/w_i < R_j^t/w_j$ since $\max_q a_{iq}^t = \max_q a_{jq}^t = 1$.
\end{proof}

\begin{corollary}\label{cor:mult:envy}
    \cref{thm:mult:envy} also proves that \nameS is envy-free (by restricting $m=1$ and $w_i = a_{i1}^t = 1$ for every user $i$ and epoch $t$) and also that \nameW, which will be studied in \cref{sec:coalitions}, is envy-free (by restricting $m=1$ and $a_{i1}^t = 1$ for every user $i$ and epoch $t$).
\end{corollary}
\section{Generalizations in Single Resource Settings}
\label{sec:gen}

\subsection{Coalitions and Dynamic Weighted Max-min Fairness}
\label{sec:coalitions}

In this section we present a generalized version of our previous results: we extend the results of \cref{sec:single} to \nameW and when users collude to increase their overall utility.

\medskip
\noindent\textbf{Notation.}
We begin with the necessary definitions for this section.

\paragraph{Dynamic weighted Max-min Fairness}
\nameW is a special case of \nameM where there is only one resource, $m=1$, i.e., every epoch $t$ the following problem is solved
\begin{align}\label{eq:gen:algo}
\begin{split}
    \textrm{choose }\;\; & r_1^t, r_2^t, \ldots, r_n^t \\
    \textrm{applying max-min fairness on }\;\; & \frac{R_1^{t-1} + r_1^t}{w_1},\; \frac{R_2^{t-1} + r_2^t}{w_2},\; \ldots,\; \frac{R_n^{t-1} + r_n^t}{w_n} \\
    \textrm{given the constraints }\;\; &
    \sum_{i\in [n]} r_i^t \le \mathcal{R}\;\;\;\textrm{ and } \;\;\;
    \min\left\{d_i^t, \a \frac{\mathcal R w_i}{\sum_k w_k}\right\}\le r_i^t \le d_i^t, \forall i
\end{split}
\end{align}
where $\a\in[0,1]$ and $w_1,\dots,w_n$ are positive numbers, similar to the definition of \nameM. Similar to previous sections, for a given $\a$, we denote $g_i(d_i^t) = \min\left\{d_i^t, \a\mathcal R w_i/\sum_k w_k\right\}$ the \textit{guarantee} of user $i$ in epoch $t$. It is easy to see that \nameW satisfies $\a$-sharing incentives.

\begin{theorem}\label{thm:gen:si}
    When every user is guaranteed an $\a$-fraction of their fair share, \nameW satisfies $\a$-sharing incentives.
\end{theorem}

\paragraph{Coalitions}
When users form coalitions they try to increase the sum of their utilities by each member of the coalition deviating. We bound that increase, i.e., if the set $I\subset [n]$ of users forms a coalition and report demands $\{\hat d_i^t\}_{i\in I,t}$ instead of $\{d_i^t\}_{i\in I,t}$, then for some $\g\ge 1$ and for all $t$ we want to prove that $\sum_{i\in I}\hat U_i^t \le \g \sum_{i\in I}U_i^t$. Proofs of the results in this section are included in \cref{sec:app:coalitions}.

\subsubsection{Over-reporting and Incentive Compatibility Upper Bound}

Analogously to \cref{thm:single:no_over_report} users have nothing to gain in \nameW by over-reporting, even when colluding.

\begin{restatable}{theorem}{generalover}\label{thm:gen:no_over_report}
    Let $I\subset [n]$ be a set of users that form a coalition. Then, for any value of $\a\in[0,1]$ used by \nameW, the users in $I$ have nothing to gain over-reporting their demand.
\end{restatable}

Next we show a generalized version of \cref{thm:single:upper}. More specifically, that users in the coalition cannot misreport their demands to increase their total utility by a factor more than $2$ and if there is no coalition ($I = \{i\}$), user $i$ cannot increase her utility by a factor larger than $1+\max_{j\neq i}\frac{w_i}{w_i+w_j}$ (which is strictly less than $2$). We should note that this result is much different than the results of \cref{sec:mult}: using the notation of that section, here we have that $\rho_i = \max_{j\ne i}\{w_i/w_j\}$. Even though $1+\rho_i$ is possibly unbounded, the incentive compatibility ratio here is at most $2$.

\begin{restatable}{theorem}{generalupper}\label{thm:gen:upper}
    Let $I\subset [n]$ be a set of users that form a coalition and $w_1,\ldots,w_n$ be any weights, according to which \nameW allocates resources. Then, for any $\a\in[0,1]$, any deviation of the users in $I$, and any epoch $t$ it holds that
    \begin{equation*}
        \sum_{i\in I}\hat U_i^t \le
        2 \sum_{i\in I}U_i^t
    \end{equation*}
    Additionally, when $I=\{i\}$ for any user $i$, it holds that $
        \hat U_i^t \le
        \left(1+\max\limits_{j\neq i}\dfrac{w_i}{w_i+w_j}\right) U_i^t
    $.
\end{restatable}

Note that if $I = \{i\}$ and users have the same weights the above theorem is similar to \cref{thm:single:upper}. To prove \cref{thm:gen:upper} we prove \cref{lem:gen:aux_bound}, a generalization of \cref{lem:single:aux_bound} bounding the over-allocation of all users, where the $r_1^t$ in the right hand side of the inequality is replaced with $\sum_{i\in I} r_i^t$.
\subsection{Getting more utility multiple times}
\label{sec:time}

In this section we study what happens when user $1$ deviates and gets more utility over an extended time period, or multiple times. More specifically, assuming that there are alternating intervals where either $\hat R_1 > R_1$ or $\hat R_1 < R_1$ we study the length of those intervals and the duration between them. We first make the following definitions:
\begin{itemize}
    \item For $\ell=0,1,2,\ldots$ let $s_\ell$ be distinct and ordered times (i.e., $s_{\ell-1} < s_\ell$) when user $1$ \textit{begins} having more resources by misreporting, i.e., $\hat R_1^{s_\ell-1} \leq R_1^{s_\ell-1}$ and $\hat R_1^{s_\ell} > R_1^{s_\ell}$.
    \item For $\ell=0,1,2,\ldots$ let $e_\ell$ be the first time after epoch $s_\ell$ when user $1$ \textit{begins} having less resources by misreporting, i.e., $\hat R_1^{e_\ell-1} \geq R_1^{e_\ell-1}$ and $\hat R_1^{e_\ell} < R_1^{e_\ell}$.
\end{itemize}

Note that $0 < s_0 < e_0 < s_1 < e_1 < \ldots$ by definition. Using the above notation we prove that if during every interval $[s_\ell,e_\ell]$ user $1$ got a factor of $\g$ more resources in epoch $t_\ell$ for some $t_\ell\in [s_\ell,e_\ell]$ by misreporting, then $t_\ell$ cannot be much larger than $s_\ell$ and that also $t_\ell$ scales exponentially with $\ell$. The proof of the theorem is presented in \cref{sec:app:time}.

\begin{restatable}{theorem}{repeat}\label{thm:time:repeat}
    Assume that for every $t$, $R_1^t\in \Theta(t)$ and for every $\ell=0,1,\ldots$ there exists an epoch $t_\ell\in [s_\ell, e_\ell)$ for which $\hat R_1^{t_\ell} \ge \g R_1^{t_\ell}$, for some $\g>1$. Then, in \nameS for any $\a\in[0,1]$, any $\ell=0,1,\ldots$, and any $t_\ell\in [s_\ell,e_\ell)$ such that $\hat R_1^{t_\ell} \ge \g R_1^{t_\ell}$, it holds that
    \begin{equation*}
        t_\ell = O(s_\ell)
        \;\;\textrm{ and }\;\;
        t_\ell = \left(\frac{2-\g}{3-2\g}\right)^\ell \Omega(t_0)
    \end{equation*}
\end{restatable}

\printbibliography

\appendix
\section{Deferred Proofs of Section \ref{sec:single}}
\label{sec:app:single}

In this section we prove \cref{thm:single:lower}, the lower bound for \nameS.

\counterexample*

\begin{proof}
    We are going to prove the theorem for $\a = 0$. We can do that without loss of generality because for a fixed $\mathcal R$ we can add users with zero demands (that do not affect the allocation) to make the guarantee $g(d_i^t) = \min\{d_i^t, \a\mathcal R/n\}$ tend to $0$ for any $\a\in[0,1]$.
    
    The example in which a user (here called Alice) can deviate to increase her allocation by a factor of almost $\sqrt 2$ is the following (the details of the example are also in \cref{tab:single:counter_sqrt_2:1,tab:single:counter_sqrt_2:2}).
    \begin{itemize}
        \item There are $1+m+k$ users: Alice, users $B_1$ to $B_m$, and users $C_1$ to $C_k$.
        \item The epochs are divided into 4 phases:
        \begin{enumerate}
            \item In the first phase there are $m$ epochs. For $i\in [m]$, in the $i$-th of these epochs only user $B_i$ has a demand, and she demands and gets an allocation of $F_i$, which will be defined later\footnote{We can assume that in every epoch there is a different total amount of resources without loss of generality: we can split an epoch into $x$ sub-epochs where in each there are $1/x$ available resources.}. We are going to assume that $F_{i+1} \ge F_i$ and also define $f_i = F_i - F_{i-1}$ with $F_0 = 0$.
            
            \item In the second phase there are again $m$ epochs. For $i\in [m]$, in the $i$-th of these epochs there are $f_i$ resources and only users Alice and $B_i$ have demand, each demanding all the resources.
            
            \item In the third phase there are $k$ epochs. For $i\in [k]$, in the $i$-th of these epochs there are $F_m$ resources and only users Alice and $C_i$ have demand, each demanding all the resources.
            
            \item In the fourth phase there are $m$ epochs. In the 1st of these epochs there are $f_m$ resources and only Alice and user $B_m$ have demands, demanding all the resources. For $i\in[m-1]$, in the $(m-i+1)$-th of these epochs there are 
            \begin{equation}\label{eq:single:app:1}
                F_m-F_i + \frac{f_m}{2}
                \;=\;
                (F_{i}+f_{i})-(F_{i+1}+f_{i+1})
            \end{equation}
            resources (the above equation defines $F_i$) and only users Alice and $B_i$ have demand, each demanding all the resources.
        \end{enumerate}
    \end{itemize}

    \begin{table}[h]
    \centering
    \begin{tabular}{c||cc|ccc|ccc|cc}
    \shortstack{{ }\\Users\\\phantom{A}} &
        \multicolumn{2}{c|}{
            \shortstack{Total allocations\\after phase 1\\{ }}
        } &
        \multicolumn{3}{c|}{
            \shortstack{Epoch $i$, phase $2$,\\Alice and $B_i$,\\$i\in[m]$}
        }&
        \multicolumn{3}{c|}{
            \shortstack{Epoch $i$, phase $3$,\\Alice and $C_i$,\\$i\in[k]$}
        }&
        \multicolumn{2}{c}{
            \shortstack{Total allocations\\after phase 3\\{ }}
        } \\\hline

    Alice &
        \cblue $0$ & \cred $0$ &
        $f_i$ & \cblue $f_i$ & \cred $0$ &
        $F_m$ & \cblue $0$ & \cred $F_m\,2^{-i}$ &
        \cblue $F_m$ & \cred $F_m\left(1-2^{-k}\right)$
        \\

    $B_i$, $i\in [m]$ &
        \cblue $F_i$ & \cred $F_i$ &
        $f_i$ & \cblue $0$ & \cred $f_i$ &
        $0$ & \cblue $0$ & \cred $0$ &
        \cblue $F_i$ & \cred $F_i + f_i$
        \\

    $C_i$, $i\in [k]$ &
        \cblue $0$ & \cred $0$ &
        $0$ & \cblue $0$ & \cred $0$ &
        $F_m$ & \cblue $F_m$ & \cred $F_m(1-2^{-i})$ &
        \cblue $F_m$ & \cred $F_m(1-2^{-i})$
        \\

    \end{tabular}
    \caption{The allocation of resources for the first 3 phases. The black numbers denote the users' demands, the blue numbers are the users' allocations in each epoch when Alice is truthful, and the red numbers are the allocations in each epoch when Alice misreports a demand of $0$ for every epoch of phase $2$. The total resources in each epoch is equal to the maximum demand of any user.}
    \label{tab:single:counter_sqrt_2:1}
    \end{table}

    \begin{table}[h]
    \centering
    \setlength{\tabcolsep}{4.6pt}
    \begin{tabular}{c||cc|cc|c|cc|c|cc}
    Users &
        \multicolumn{2}{c|}{
            \shortstack{Epoch $1$\\phase $4$}
        } &
        \multicolumn{2}{c|}{
            \shortstack{Epoch $2$\\phase $4$}
        } &
        \dots &
        \multicolumn{2}{c|}{
            \shortstack{Epoch $m-i+1$\\phase $4$}
        } &
        \dots &
        \multicolumn{2}{c}{
            \shortstack{Epoch $m$\\phase $4$}
        } \\\hline

    Alice &
        $f_m$ & \cblue $\nicefrac{f_m}{2}$ &
        $\substack{F_m+\nicefrac{f_m}{2}\\-F_{m-1}}$ & \cblue $0$ &
        \dots &
        $\substack{F_m+\nicefrac{f_m}{2}\\-F_i}$ & \cblue $0$ &
        \dots &
        $\substack{F_m+\nicefrac{f_m}{2}\\-F_1}$ & \cblue $0$
        \\

    $B_1$ &
        $0$ & \cblue $0$ &
        $0$ & \cblue $0$ &
        \dots &
        $0$ & \cblue $0$ &
        \dots &
        $\substack{F_m+\nicefrac{f_m}{2}\\-F_1}$ & \cblue $\substack{F_m+\nicefrac{f_m}{2}\\-F_1}$
        \\

    \dots & \dots & \dots & \dots & \dots & \dots & \dots & \dots & \dots & \dots & \dots\\
    
    $B_i$ &
        $0$ & \cblue $0$ &
        $0$ & \cblue $0$ &
        \dots &
        $\substack{F_m+\nicefrac{f_m}{2}\\-F_i}$ & \cblue $\substack{F_m+\nicefrac{f_m}{2}\\-F_i}$ &
        \dots &
        $0$ & \cblue $0$
        \\
    
    \dots & \dots & \dots & \dots & \dots & \dots & \dots & \dots & \dots & \dots & \dots\\

    $B_{m-1}$ &
        $0$ & \cblue $0$ &
        $\substack{F_m+\nicefrac{f_m}{2}\\-F_{m-1}}$ & \cblue $\substack{F_m+\nicefrac{f_m}{2}\\-F_{m-1}}$ &
        \dots &
        $0$ & \cblue $0$ &
        \dots &
        $0$ & \cblue $0$
        \\
    
    $B_m$ &
        $f_m$ & \cblue $\nicefrac{f_m}{2}$ &
        $0$ & \cblue $0$ &
        \dots &
        $0$ & \cblue $0$ &
        \dots &
        $0$ & \cblue $0$
        \\\hline\hline
    
    Alice &
        $f_m$ & \cred $f_m$ &
        $\substack{(F_{m-1}+f_{m-1})\\-(F_{m}+f_{m})}$ & \cred $\substack{(F_{m-1}+f_{m-1})\\-(F_{m}+f_{m})}$ &
        \dots &
        $\substack{(F_{i}+f_{i})\\-(F_{i+1}+f_{i+1})}$ & \cred $\substack{(F_{i}+f_{i})\\-(F_{i+1}+f_{i+1})}$ &
        \dots &
        $\substack{(F_{1}+f_{1})\\-(F_{2}+f_{2})}$ & \cred $\substack{(F_{1}+f_{1})\\-(F_{2}+f_{2})}$
        \\

    $B_1$ &
        $0$ & \cred $0$ &
        $0$ & \cred $0$ &
        \dots &
        $0$ & \cred $0$ &
        \dots &
        $\substack{(F_1+f_1)\\-(F_2+f_2)}$ & \cred $0$
        \\

    \dots & \dots & \dots & \dots & \dots & \dots & \dots & \dots & \dots & \dots & \dots\\
    
    $B_i$ &
        $0$ & \cred $0$ &
        $0$ & \cred $0$ &
        \dots &
        $\substack{(F_{i}+f_{i})\\-(F_{i+1}+f_{i+1})}$ & \cred $0$ &
        \dots &
        $0$ & \cred $0$
        \\
    
    \dots & \dots & \dots & \dots & \dots & \dots & \dots & \dots & \dots & \dots & \dots\\

    $B_{m-1}$ &
        $0$ & \cred $0$ &
        $\substack{(F_{m-1}+f_{m-1})\\-(F_{m}+f_{m})}$ & \cred $0$ &
        \dots &
        $0$ & \cred $0$ &
        \dots &
        $0$ & \cred $0$
        \\
    
    $B_m$ &
        $f_m$ & \cred $0$ &
        $0$ & \cred $0$ &
        \dots &
        $0$ & \cred $0$ &
        \dots &
        $0$ & \cred $0$
        \\

    \end{tabular}
    \caption{The allocation of resources for phase 4. The top part of the table is for when Alice is truthful, and the bottom part is for when Alice misreports a demand of $0$ for every epoch of phase $2$.
    The total resources in each epoch is equal to the maximum demand of any user.}
    \label{tab:single:counter_sqrt_2:2}
    \end{table}
    
    If Alice is truthful she will end up with $F_m + \nicefrac{f_m}{2}$ resources after phase $4$. If she misreports a demand of $0$ in every epoch of phase $2$ she will get an allocation of $F_1 + f_1 - F_m\,2^{-k}$ after phase $4$ (both cases are depicted in \cref{tab:single:counter_sqrt_2:1,tab:single:counter_sqrt_2:2}). Solving the recursive equations \eqref{eq:single:app:1} and taking $m\to\infty$ and $k\to\infty$ we get that
    \begin{equation*}
        \frac{F_1 + f_1 - F_m\,2^{-k}}{F_m + \nicefrac{f_m}{2}}
        \to
        \sqrt 2
    \end{equation*}
    which is the desired incentive compatibility ratio. For completeness, the solution for $F_i$ when we set $F_m + f_m/2 = 1$ (for normalization) is
    \begin{equation*}
        F_i
        \;=\;
        \frac{
            \left(\sqrt{2} + 1\right)^{m+1} \left(
                1 - \left( \frac{2 - \sqrt 2}{2}\right)^i
            \right)
            +
            \left(\sqrt 2 - 1\right)^{m+1} \left(
                1 - \left(\frac{2 + \sqrt 2}{2}\right)^i
            \right) 
        }{
            \left(\sqrt 2 - 1\right)^{m+1}
            +
            \left(\sqrt 2 + 1\right)^{m+1}
        }
    \end{equation*}
    One can check that our initial conditions are true: for any $m\ge 1$, $F_{i+1} \ge F_i$ and the quantities in \eqref{eq:single:app:1} are positive. 
\end{proof}

\section{Deferred Proofs of Section \ref{sec:mult}}
\label{sec:app:mult}

\subsection{Incentives when users have zero ratios}

In this section we include the discussion about \nameM when some users have zero ratios. This special case is much different, even when there is only one epoch. To illustrate this, consider the example of \cite{DBLP:conf/sigecom/ParkesPS12}: there are two resources and three users with infinite demands and ratios $(1,\e)$, $(\e,1)$, and $(\e,1)$, for a small but positive $\e$. If the total amount available of every resource is $\mathcal R = 1$, then DRF gives every user an allocation of almost $1/2$. If instead $\e=0$, it is easy to check that the allocations are much different: the first user receives an allocation of $1$, while the other two users still receive an allocation of $1/2$.

For \nameM, we prove that a user can {\em over-report} her demand to increase her allocation, even when the ratios stay fixed over time. In the next section, we will show that this is no longer possible once all resources are used by all users, even if the ratios are very different and change over time.

\begin{theorem}\label{thm:mult:unbounded_rat}
    For any $\a\in[0,1]$, there is an instance were all the users have the same weights ($w_i = w_j$ for every users $i,j$), there are $m$ resources, and users have static ratios some of which are zero where a user can increase her utility by a factor of $\Theta(m)$ by only over-reporting her demand.
\end{theorem}

\begin{proof}
    Similar to the proof of \cref{thm:single:lower}, we are going to prove the theorem for $\a = 0$. We can do that without loss of generality because for a fixed $\mathcal R$ we can add users with zero demands (that do not affect the allocation) to make the guarantee $g(d_i^t) = \min\big\{d_i^t, \a\frac{\mathcal R w_i}{\sum_k w_k}\big\}$ tend to $0$ for any $\a\in[0,1]$.

    \Cref{tab:mult:unbounded} presents the example graphically. In detail, the example has $n^2 + m - 1$ users and $m\ge 3$ resources:
    \begin{itemize}
        \item Alice uses every resource with a ratio of $1$, except for the second resource which she does not use.
        \item There are $n^2$ copies of a user, Bob, where each uses the first resource with a ratio of $1$, the second one with a ratio of $1-\frac{1}{n^3}$, and does not want any other resource.
        \item User $i$, for all $i\in [m-2]$, uses the second resource with a ratio of $\frac{1}{n (m-2)}$, the $(i+2)$-th one with a ratio of $1$, and does not use any other resources.
        \item In epoch 1 the total available amount of every resource is $1$. On the other epochs, $2,3,\ldots, m-1$, the total available amount for each resource is $\frac{1}{n^2}$\footnote{We can make that assumption w.l.o.g. because we can split the first epoch where there is $1$ amount of every resource into $n^2$ epochs where the available amount of every resource is $n^{-2}$ every round.}.
    \end{itemize}

    \begin{table}[h]
        \centering
        \setlength{\tabcolsep}{3.4pt}
        \begin{tabular}{c||ccccccc|ccc|ccc|ccc|ccc|c}
            Users & \multicolumn{7}{c|}{Ratios for resources $1,\ldots,m$} & \multicolumn{3}{c|}{Epoch 1} & \multicolumn{3}{c|}{Epoch 2} & \multicolumn{3}{c|}{Epoch 3} & \multicolumn{3}{c|}{Epoch 4} & ...\\\hline
            Alice &
                $1$ & $0$ & $1$ & $1$ & $1$ & ... & $1$ & 
                0 & \cblue{0} & \cred{$\frac{1}{n^2+1}$} &
                $\frac{1}{n^2}$ & \cblue{$\frac{1}{n^2}$} & \cred{$\frac{1}{n^2}$} & 
                $\frac{1}{n^2}$ & \cblue{$\frac{1}{2n^2}$} & \cred{$\frac{1}{n^2}$} & 
                $\frac{1}{n^2}$ & \cblue{$\frac{1}{4n^2}$} & \cred{$\frac{1}{n^2}$} & ...\\
            $n^2$ Bobs & 
                $1$ & $1-\frac{1}{n^3}$ & $0$ & $0$ & $0$ & ... & $0$ & 
                1 & \cblue{$\frac{1}{n^2}$} & \cred{$\frac{1}{n^2+1}$} &
                $0$ & \cblue{$0$} & \cred{$0$} &
                $0$ & \cblue{$0$} & \cred{$0$} &
                $0$ & \cblue{$0$} & \cred{$0$} & ... \\
            User $1$ & 
                $0$ & $\frac{1}{n (m-2)}$ & $1$ & $0$ & $0$ & ... & $0$ & 
                1 & \cblue{$\frac{1}{n^2}$} & \cred{$\frac{n+1}{n^2+1}$} &
                $\frac{1}{n^2}$ & \cblue{0} & \cred{0} &
                0 & \cblue{0} & \cred{0} & 
                0 & \cblue{0} & \cred{0} & ... \\
            User $2$ &
                $0$ & $\frac{1}{n (m-2)}$ & $0$ & $1$ & $0$ & ... & $0$ &
                1 & \cblue{$\frac{1}{n^2}$} & \cred{$\frac{n+1}{n^2+1}$} &
                0 & \cblue{0} & \cred{0} &
                $\frac{1}{n^2}$ & \cblue{$\frac{1}{2n^2}$} & \cred{0} &
                0 & \cblue{0} & \cred{0} & ... \\
            User $3$ &
                $0$ & $\frac{1}{n (m-2)}$ & $0$ & $0$ & $1$ & ... & $0$ & 
                1 & \cblue{$\frac{1}{n^2}$} & \cred{$\frac{n+1}{n^2+1}$} &
                0 & \cblue{0} & \cred{0} & 
                0 & \cblue{0} & \cred{0} & 
                $\frac{1}{n^2}$ & \cblue{$\frac{3}{4n^2}$} & \cred{0} & ... \\
            ... & ... & ... & ... & ... & ... & ... & ... & ... & ... & ... & ... & ... & ... & ... & ... & ... & ... & ... & ... & ... \\
            User $m-2$ &
                $0$ & $\frac{1}{n (m-2)}$ & $0$ & $0$ & $0$ & ... & $1$ & 
                1 & \cblue{$\frac{1}{n^2}$} & \cred{$\frac{n+1}{n^2+1}$} &
                0 & \cblue{0} & \cred{0} & 
                0 & \cblue{0} & \cred{0} & 
                0 & \cblue{0} & \cred{0} & ... \\
        \end{tabular}
        \caption{For every epoch the black numbers denote the users' demands, the blue numbers are the users' allocation when Alice is truthful, and the red numbers are the allocations when Alice over-reports her demand in epoch 1 by asking for 1 resource. The available amount of every resource in epoch 1 is 1, while on the other epochs it is $\frac{1}{n^2}$.}
        \label{tab:mult:unbounded}
    \end{table}

    If Alice lies on the first epoch and requests resources instead of requesting zero, she can take resources from every Bob, which increases the allocation of the other users by a factor of $\Theta(n)$. This entails that in epochs $2$, $3$, ..., $m-1$, instead of Alice getting allocated $\frac{1}{n^2}$, $\frac{1}{2n^2}$, ..., $\frac{1}{2^{m-3}n^2}$, respectively, she gets $\frac{1}{n^2}$ every round. Note that this is true if the following inequality is true, which guarantees that when Alice is untruthful the resulting allocation in epochs $2$ to $m-1$ is max-min fair:
    \begin{equation*}
        \frac{1}{n^2+1} + \frac{m-2}{n^2} \leq \frac{n+1}{n^2+1}
    \end{equation*}
    which is true if $m\leq 2 + \frac{n^3}{n^2+1}$. This allows us to set $m = \Theta(n)$. In total, Alice gets an allocation of $\frac{2}{n^2}\left(1-\frac{1}{2^{m-2}}\right) \leq \frac{2}{n^2}$ if she is truthful and a (useful) allocation of $\frac{m-2}{n^2}$ if she is untruthful, i.e., she increases her utility by a factor of at least $\frac{m-2}{2}$. This proves the theorem.
\end{proof}

\subsection{Deferred Proofs of Section \ref{sec:mult:nonzero}}

First we restate and prove \Cref{lem:mult:more_less}.

\multMoreLess*

\begin{proof}
    We notice that
    \begin{itemize}
        \item Because $\hat r_i^t > g(\hat d_i^t)$ (since $\hat r_i^t > \bar r_i^t \ge g(\bar d_i^t) \ge g(\hat d_i^t)$) and $\hat a_{jq}^t > 0$ for all $q$, decreasing $\hat r_i^t$ would free a positive amount of every resource.
        \item Because $\hat r_j^t < \hat d_j^t$ it is feasible to increase $\hat r_j^t$ in order to increase $\hat R_j^t$.
    \end{itemize}
    This implies that $\hat R_i^t/w_i \le \hat R_j^t/w_j$; otherwise it would have been more fair to give some of the resources user $i$ got to user $j$.
    With the analogous inverse argument (we can increase $\bar r_i^t$ by decreasing $\bar r_j^t$) we can prove that $\bar R_i^t/w_i \ge \bar R_j^t/w_j$. This completes the proof.
\end{proof}

Finally we prove the lower bound for the incentive compatibility ratio of \nameM.

\multLower*

\begin{proof}
    We are going to prove the theorem for $\a = 0$. We can make the assumption that $\a = 0$ without loss of generality because for a fixed $\mathcal R$ we can add users with zero demands and same ratios and weights as some user other that user $1$ (thus not affecting the allocation nor the definition of $\rho_1$) to make the guarantee $g(d_i^t) = \min\big\{d_i^t, \a\frac{\mathcal R w_i}{\sum_k w_k}\big\}$ tend to $0$ for any $\a\in[0,1]$.

    There are $2$ epochs, $2$ resources, and $4$ groups of users where the users in each group have the same ratios and demands. Groups $1$ and $2$ have one user each, user $1$ and user $2$ respectively, while groups $3$ and $4$ have $n_1$ and $n_2$ users, respectively.
    
    User $1$ is the one who will under-report her demand to increase her utility and w.l.o.g. we assume that her weight is $w_1 = 1$. All other users have weight $w > 0$ (this corresponds to the weight $w_2$ in the theorem's statement).
    Their ratios are depicted in \cref{tab:mult:counter}, along with a summary of the whole example. We assume that $\d$ and $\e$ are fixed. We note that
    \begin{equation*}
        \rho_1
        \;=\;
        \max_{k\ne 1, q}\frac{w_1 a_{1q}}{w_k a_{kq}}
        \;=\;
        \max\left\{
            \frac{1}{w\e},
            \frac{1}{w},
            \frac{1}{w\e},
            \frac{\d}{w},
            \frac{\d}{w\e},
            \frac{\d}{w}
        \right\}
        \;=\;
        \frac{1}{w\e}
    \end{equation*}
    as needed by the theorem.

    \begin{table}[h]
    \centering
    \setlength{\tabcolsep}{4pt}
    \begin{tabular}{c||c|c||ccc|ccc}
        & Res. 1 & Res. 2  &
            \multicolumn{3}{c|}{Ep. 1 $\Big(\mathcal R^1 = 1+\frac{n_1 w}{1 + w\e}\Big)$} &
            \multicolumn{3}{c}{Ep. 2 $\Big(\mathcal R^2 = \frac{\d}{w\e}+\frac{n_2 \d}{\e(w+\d)}\Big)$} \\\hline
        User $1$ ($w_1=1$) & $1$ & $\d$ &
            $\infty$ & \cblue $\frac{1}{1+w\e}$ & \cred $0$ &
            $\infty$ & \cblue $\frac{\d}{w\e(w+\d)}$ & \cred $\frac{1}{w\e}$ \\
        User $2$ ($w_2=w$)& $\e$ & $1$ &
            $\infty$ & \cblue $\frac{w}{1+w\e}$ & \cred $\frac{1}{\e}$ &
            $\infty$ & \cblue $\frac{\d}{\e(w+\d)}$ & \cred $0$ \\
        Group $3$ ($n_1$ users, $w_3=w$) & $1$ & $\e$ &
            $\frac{w}{1+w\e}$ & \cblue $\frac{w}{1+w\e}$ & \cred $\frac{w}{1+w\e}$ &
            $0$ & \cblue $0$ & \cred $0$ \\
        Group $4$ ($n_2$ users, $w_4=w$) & $\e$ & $1$ &
            $0$ & \cblue $0$ & \cred $0$ &
            $\frac{\d}{\e(w+\d)}$ & \cblue $\frac{\d}{\e(w+\d)}$ & \cred $\frac{\d}{\e(w+\d)}$ \\
    \end{tabular}
    \caption{Summary of the example of \cref{thm:mult:lower}.
    The 2nd and 3rd columns depict the users' ratios, where $0<\d\le\e<1$.
    The last two columns for each epoch show the total available resources for that round in the parentheses of the first row, the demands of each group (black numbers), their allocations when the user in group $1$ is truthful (blue numbers), and their allocations when the user in group $1$ misreports a demand of $0$ in epoch $1$ (red numbers).}
    \label{tab:mult:counter}
    \end{table}
    
    \medskip
    \noindent\textbf{Epoch 1.}
    In the first epoch, users $1$ and $2$ demand $\infty$ resources, group $3$ demand $\frac{w}{1+w\e}$, and group $4$ demand $0$. The total amount available for every resource is $1 + \frac{n_1 w}{1+w\e}$.

    We first show that with these demands, user $1$ gets $r_1^1 = \frac{1}{1 + w\e}$, while user $2$ and users in group $3$ get an allocation of $\frac{w}{1 + w\e}$. Note that this is fair, since every user with positive demand gets an allocation proportional to their weight. Because this allocation is fair, to prove its validity we need that one resource is saturated and the other is not over-used. The total amount of resource $1$ used is
    \begin{equation*}
        \frac{1}{1+w\e} +
        \e\frac{w}{1+w\e} +
        n_1 \frac{w}{1+w\e} =
        1 + \frac{n_1 w}{1+w\e}
        \;=\;
        \mathcal R^1
    \end{equation*}
    meaning that resource $1$ is saturated, proving that we cannot increase the users' allocations. The total amount of resource $2$ used is
    \begin{equation*}
        \d \frac{1}{1+w\e} +
        \frac{w}{1+w\e} +
        n_1 \e \frac{w}{1+w\e}
        \;=\;
        \frac{\d + w}{1 + w\e} + n_1\frac{w \e}{1 + w\e}
        \;\le\;
        1 + n_1\frac{w}{1+w\e}
        \;=\;
        \mathcal R^1
    \end{equation*}
    where the inequality holds because $\e<1$ and we assume that $n_1$ is large enough. This proves that we do not over-use resource $2$.

    Now we will prove that if the user in group $1$ misreports her demand and asks for $0$ instead, then the allocations are as follows: user $1$ gets $\hat r_1^1 = 0$, user $2$ gets an allocation of $\frac{1}{\e}$ and users in group $3$ get an allocation of $\frac{w}{1+w\e}$. Note that this is fair, since they have the same weight and users in group $3$ get their demand, while user $2$ gets more than. The total amount of resource $1$ used is
    \begin{equation*}
        0 +
        \e\frac{1}{\e} +
        n_1 \frac{w}{1+w\e} =
        1 + \frac{n_1 w}{1+w\e}
        \;=\;
        \mathcal R^1
    \end{equation*}
    meaning that resource $1$ is saturated, proving that we cannot increase the users' allocations. The total amount of resource $2$ used is
    \begin{equation*}
        0 +
        \frac{1}{\e} +
        n_1 \e \frac{w}{1+w\e} \le
        1 + n_1\frac{w}{1+w\e}
        \;=\;
        \mathcal R^1
    \end{equation*}
    where the inequality holds because $\e < 1$ and we assume that $n_1$ is large enough. This proves that we do not over-use resource $2$.

    \medskip
    \noindent\textbf{Epoch 2.}
    In the second epoch, users $1$ and $2$ demand $\infty$ resources, group $3$ demand $0$, and group $4$ demand $\frac{w\d}{\e(w + \d)}$. The total amount available for every resource is $\frac{\d}{w\e} + \frac{n_2\,\d}{\e(w+\d)}$.

    We first show that with these demands, when user $1$ was truthful, user $1$ gets $r_1^2 = \frac{\d}{w\e(w+\d)}$, while user $2$ and users in group $4$ get an allocation of $\frac{\d}{\e(w+\d)}$. Note that this is fair, since user $2$'s total allocation is the same as user $1$'s times $w$ and users in group $4$ get their demand which is less than the total allocation of user $2$ with which they have the same weight. The total amount of resource $1$ used is
    \begin{equation*}
        \frac{\d}{w\e(w+\d)} +
        \e\frac{\d}{\e(w+\d)} +
        n_2 \e \frac{\d}{\e(w+\d)}
        \;\le\;
        \frac{\d}{w\e} + n_2\frac{\d}{\e(w+\d)}
        \;=\;
        \mathcal R^2
    \end{equation*}
    where the inequality holds because $\e < 1$ and by assuming that $n_2$ is large enough. This proves that resource $1$ is not over-used. The total amount of resource $2$ used is
    \begin{equation*}
        \d \frac{\d}{w\e(w+\d)} +
        \frac{\d}{\e(w+\d)} +
        n_2 \frac{\d}{\e(w+\d)}
        \;=\;
        \frac{\d}{w\e} + n_2\frac{\d}{\e(w+\d)}
        \;=\;
        \mathcal R^2
    \end{equation*}
    meaning that resource $1$ is saturated, proving that we cannot increase the users' allocations.

    Now we will prove that when the user $1$ misreported in the first epoch, the allocations are as follows: user $1$ gets $\hat r_1^2 = \frac{1}{w\e}$, user $2$ gets $0$, and group $4$ gets $\frac{\d}{\e(w+\d)}$. Note that this is fair; the total resource of user $1$ and user $2$ are now $\frac{1}{w\e}$ and $\frac{1}{\e}$ (which are proportional to users' weights) and users in group $4$ get their demand which is not more than what user $2$ gets in total who has the same weight.
    The total amount of resource $1$ used is
    \begin{equation*}
        \frac{1}{w\e} +
        \e\, 0 +
        n_2 \e \frac{\d}{\e(w+\d)}
        \;\le\;
        \frac{\d}{w\e} + n_2\frac{\d}{\e(1+\d)}
        \;=\;
        \mathcal R^2
    \end{equation*}
    where the inequality holds because $\e < 1$ and by assuming that $n_2$ is large enough. This proves that resource $1$ is not over-used. The total amount of resource $2$ used is
    \begin{equation*}
        \d \frac{1}{w\e} +
        0 +
        n_2 \frac{\d}{\e(w+\d)}
        \;=\;
        \frac{\d}{w\e} + n_2\frac{\d}{\e(w+\d)}
        \;=\;
        \mathcal R^2
    \end{equation*}
    meaning that resource $1$ is saturated, proving that we cannot increase the users' allocations.

    This proves that the misreporting of user $1$ increases her resources by
    \begin{equation}\label{eq:mult:87}
        \frac{\hat r_1^1 + \hat r_1^2}{r_1^1 + r_1^2}
        \;=\;
        \frac{0+\frac{1}{w\e}}
             {\frac{1}{1+w\e} + \frac{\d}{w\e(w+\d)}}
        \;\overset{\d\to 0}{\to}\;
        1 + \frac{1}{w\e} 
        \;=\;
        1 + \rho_1
    \end{equation}
    the above proves that as $\d\to 0$ we get an incentive compatibility ratio of $1+\rho_1$, as needed.
\end{proof}

It is worth pointing out what we mentioned in \cref{sec:mult}: if we want a lower bound that depends on $\rho = \max_i \rho_i$ instead of $\rho_1$, we get an interesting result when $w=1$: if we set $\d=\e$, making $\rho = 1/\e$, \eqref{eq:mult:87} now proves an incentive compatibility bound of $\frac{1+1/\e}{2} = \frac{1+\rho}{2}$.
\section{Deferred Proofs of Section \ref{sec:coalitions}}
\label{sec:app:coalitions}

First we state a more specific version of \Cref{lem:mult:more_less}.

\begin{lemma}\label{lem:gen:more_less}
    Fix an epoch $t$ and the allocations of two different outcomes $\{\hat R_k^{t-1}\}_{k\in [n]}$ and $\{\bar R_k^{t-1}\}_{k\in [n]}$. Let $i,j$ be two different users. If the following conditions hold
    \begin{itemize}
        \item For $i$, $\;\bar r_i^t < \hat r_i^t$ and $\;\hat d_i^t \le \bar d_i^t$.
        \item For $j$, $\;\bar r_j^t > \hat r_j^t$ and $\;\bar d_j^t \le \hat d_j^t$.
    \end{itemize}
    then, for any $\a\in[0,1]$,
    \begin{equation*}
        \frac{\bar R_i^t}{w_i} \geq \frac{\bar R_j^t}{w_j}
        \;\;\;\textrm{ and }\;\;\;
        \frac{\hat R_i^t}{w_i} \le \frac{\hat R_j^t}{w_j}
    \end{equation*}
    and consequentially
    \begin{equation}\label{eq:gen:0}
        \frac{\hat R_i^t - \bar R_i^t}{w_i} \le \frac{\hat R_j^t - \bar R_j^t}{w_j}
    \end{equation}
\end{lemma}

Next we generalize \Cref{lem:single:aux_bound}. Note that this is much different and stronger than what is stated in \cref{lem:mult:aux_bound}; even without the assumption that users collude \cref{lem:mult:aux_bound} proves that in \nameW, for all $i,t$:
\begin{equation*}
    \hat R_i^t - \bar R_i^t \le \bar R_1^t \max_k\frac{w_i}{w_k}
\end{equation*}

\begin{lemma}\label{lem:gen:aux_bound}
    Fix an epoch $t$ and the allocations of two different outcomes $\{\hat R_k^{t-1}\}_{k\in [n]}$ and $\{\bar R_k^{t-1}\}_{k\in [n]}$. Assume that $\{\bar d_i^t\}_{i\in [n]}$ are some users' demands and that $\{\hat d_i^t\}_{i\in [n]}$ are the same demands except users in $I$, who deviate but not by over-reporting, i.e., $\hat d_i^t \le \bar d_i^t$ for $i\in I$. Then, for any $\a\in[0,1]$, it holds that
    \begin{equation}\label{eq:gen:aux}
        \sum_{k\in [n]}\left(\hat R_k^t - \bar R_k^t\right)^+
        -
        \sum_{k\in [n]}\left(\hat R_k^{t-1} - \bar R_k^{t-1}\right)^+
        \le
        \One{\sum_{k\in I}\hat d_k^t < \sum_{k\in I}\bar d_k^t} \sum_{k\in I}\bar r_k^t
    \end{equation}
\end{lemma}

\begin{proof}
    First we define $P^t=\{k\in [n]: \hat R_k^t > \bar R_k^t\}$ for all $t$. Suppose by contradiction:
    \begin{equation*}
        \sum_{k\in P^t}\left(\hat R_k^t - \bar R_k^t\right)
        -
        \sum_{k\in P^{t-1}}\left(\hat R_k^{t-1} - \bar R_k^{t-1}\right)
        >
        \One{\sum_{k\in I}\hat d_k^t < \sum_{k\in I}\bar d_k^t} \sum_{k\in I}\bar r_k^t
    \end{equation*}
    
    Because $\sum_{k\in P^{t}}(\hat R_k^{t-1} - \bar R_k^{t-1})\le\sum_{k\in P^{t-1}}(\hat R_k^{t-1} - \bar R_k^{t-1})$, the above inequality implies
    \begin{equation}\label{eq:gen:11}
        \sum_{k\in P^t}\left(\hat r_k^t - \bar r_k^t\right) >
        \One{\sum_{k\in I}\hat d_k^t < \sum_{k\in I}\bar d_k^t} \sum_{k\in I}\bar r_k^t
    \end{equation}

    Because users only under-report their demand, $\hat d_k^t \le \bar d_k^t$, it holds that $\sum_k \bar r_i^t \ge \sum_k \hat r_i^t$, i.e., the total resources allocated to the users does not increase when user $1$ deviates. Combining this fact with \eqref{eq:gen:11} we get that
    \begin{equation}\label{eq:gen:12}
        \sum_{k\notin P^t}\left(\bar r_k^t - \hat r_k^t\right)
        >
        \One{\sum_{k\in I}\hat d_k^t < \sum_{k\in I}\bar d_k^t} \sum_{k\in I}\bar r_k^t
    \end{equation}

    We notice that because of \eqref{eq:gen:11}, there exists a user $i\in P^t$ for whom $\hat r_i^t > \bar r_i^t$; because of \eqref{eq:gen:12}, there exists a user $j\notin P^t$ for whom $\bar r_j^t > \hat r_j^t$. Additionally for that $j$ we can assume that $\hat d_j^t = \bar d_j^t$ because:
    \begin{itemize}
        \item If users in $I$ do not deviate then for all $k$, $\hat d_k^t = \bar d_k^t$.
        
        \item If $\sum_{k\in I}\hat d_k^t < \sum_{k\in I}\bar d_k^t$, then \eqref{eq:gen:12} implies $\sum_{k\notin P^t \cup I}\left(\bar r_k^t - \hat r_k^t\right) > 0$, i.e., $j\notin I$ and we assumed that only users in $I$ deviate.
    \end{itemize}
    
    Thus we have $\hat d_i^t \le \bar d_i^t$ (since no user over-reports), $\hat d_j^t = \bar d_j^t$, $\hat r_i^t > \bar r_i^t$, and $\hat r_j^t < \bar r_j^t$. Now \Cref{lem:gen:more_less} proves that $(\hat R_i^t - \bar R_i^t)/w_i \le (\hat R_j^t - \bar R_j^t)/w_j$. This leads to a contradiction, because $i\in P^t$ and $j\notin P^t$, i.e., $\hat R_i^t - \bar R_i^t > 0 \ge \hat R_j^t - \bar R_j^t$.
\end{proof}

Next we prove that users do not want to over-report their demand.

\generalover*

Similar to \cref{thm:single:no_over_report}, we first prove an auxiliary lemma.

\begin{lemma}\label{lem:gen:over_aux}
    Fix an epoch $T_0$ and the allocations of an outcome $\{\hat R_k^{T_0-1}\}_{k\in [n]}$. Fix another epoch $T\ge T_0$ and assume that in epochs $T_0+1,T_0+2,\ldots,T$ users in $I$ do not over-report their demand. Then, for any $\a\in[0,1]$, the users in $I$ cannot increase their utility in round $T$ by over-reporting their demand in epoch $T_0$.
\end{lemma}

\begin{proof}
    To prove the lemma we are going to create another outcome in which the over-reports in epoch $T_0$ are changed to a truthful ones and prove that this does not decrease $I$'s total utility in epoch $T$.
    
    For all $k,t$, let $\bar d_k^t = \hat d_k^t$, except for epoch $T_0$ and users $i\in I$, where $\bar d_i^{T_0}$ are the actual demands if they were over-reports: $\bar d_i^{T_0} = \min\{d_i^{T_0}, \hat d_i^{T_0}\}$ for $i\in I$. This means that the two outcomes are the same before epoch $T_0$, i.e., for all $k$, $\bar R_k^{T_0-1} = \hat R_k^{T_0-1}$ and $\bar U_k^{T_0-1} = \hat U_k^{T_0-1}$. In epoch $T_0$ users' allocations are different. First, let $I'\sub I$ be the users who over-report their demand in $T_0$, i.e., the users $i$ for whom $\bar d_i^{T_0} = d_i^{T_0} < \hat d_i^{T_0}$. Now we prove that for any user $i$ who got more resources, $\hat r_i^{T_0} > \bar r_i^{T_0}$, that user must be in $I'$ and that additional resources must be in excess of their true demand:
    \begin{equation}\label{eq:gen:91}
        \textrm{if }\; \hat r_i^{T_0} - \bar r_i^{T_0} = \hat R_i^{T_0} - \bar R_i^{T_0} > 0
        \;\textrm{ then }\; \bar r_i^{T_0} = \bar d_i^{\,T_0} \;\textrm{ and }\; i\in I'
    \end{equation}
    To prove \eqref{eq:gen:91} we study two cases:
    \begin{itemize}
        \item If $\hat r_i^{T_0} - \bar r_i^{T_0} = \hat R_i^{T_0} - \bar R_i^{T_0} > 0$ and $\bar r_i^{T_0} < \bar d_i^{\,T_0}$, then $\hat r_i^{T_0} < \hat d_i^{T_0}$. This means user $i$'s demands are not satisfied both in $\hat d^{T_0}$ and $\bar d^{T_0}$, which implies that, since $\hat r_i^{T_0} > \bar r_i^{T_0}$ and for all users $k$, $\bar d_k^{T_0} \le \hat d_k^{T_0}$, for some other user $j$, $\hat r_j^{T_0} - \bar r_j^{T_0} = \hat R_j^{T_0} - \bar R_j^{T_0} < 0$. This, because of \Cref{lem:gen:more_less} and $\bar d_j^{T_0} \le \hat d_j^{\,T_0}$, entails that $(\hat R_i^{T_0} - \bar R_i^{T_0})/w_i \le (\hat R_j^{T_0} - \bar R_j^{T_0})/w_j$, a contradiction.
        
        \item If $\hat r_i^{T_0} - \bar r_i^{T_0} = \hat R_i^{T_0} - \bar R_i^{T_0} > 0$ and $i\notin I'$, then, because user $i$ got more resources in $\hat d^{T_0}$ with the same demand, for some user $j$, $\hat r_j^{T_0} - \bar r_j^{T_0} = \hat R_j^{T_0} - \bar R_j^{T_0} < 0$, which because of \Cref{lem:gen:more_less}, $\bar d_i^{T_0} = \hat d_i^{T_0}$, and $\bar d_j^{\,T_0} \le \hat d_j^{T_0}$ entails that $(\hat R_i^{T_0} - \bar R_i^{T_0})/w_i \le (\hat R_j^{T_0} - \bar R_j^{T_0})/w_j$, a contradiction.
    \end{itemize}
    
    Because of \eqref{eq:gen:91} we see that any additional resources that users in $I$ get in $T_0$ when they over-report are in excess of their demand (because for $i\in I'$, $\bar d_i^{\,T_0}$ is the true demand). This means that $\sum_{i\in I}(\hat R_i^{T_0} - \bar R_i^{T_0})^+ = x \ge 0$ is in excess of their demand and therefore
    \begin{equation}\label{eq:gen:85}
        \sum_{i\in I}\left(\hat U_i^{T_0} - \bar U_i^{T_0}\right)
        =
        \sum_{i\in I}\left(\hat R_i^{T_0} - \bar R_i^{T_0}\right) - x \le 0
    \end{equation}
    
    Additionally, because users in $I$ do not over-report $\hat d$ or $\bar d$ in epochs $T_0+1$ to $T$, it holds that for $t\in [T_0+1, T]$ $I$' utility is the same as the resources they receives: $\bar u_i^t = \bar r_i^t$ and $\hat u_i^t = \hat r_i^t$ for $i\in I$.  This fact, combined with \eqref{eq:gen:85} proves that
    \begin{equation*}
        \forall t\in[T_0,T],\;
        \sum_{i\in I}\left(\hat U_i^t - \bar U_i^t\right)
        =
        \sum_{i\in I}\left(\hat R_i^t - \bar R_i^t\right) - x
    \end{equation*}
    
    Thus, in order for this over-reporting to be a strictly better strategy, it most hold that $\sum_{i\in I}(\hat R_i^t - \bar R_i^t) > x$. We will complete the proof by proving that the opposite holds. Since in epochs $t\in[T_0+1, T]$ it holds that $\hat d_i^t = \bar d_i^t$ for $i\in I$, we can use \cref{lem:gen:aux_bound} to sum \eqref{eq:gen:aux} for all $t\in [T_0+1, T]$ and get that
    \begin{equation*}
        \sum_k\left(\hat R_k^T - \bar R_k^T\right)^+
        -
        \sum_k\left(\hat R_k^{T_0} - \bar R_k^{T_0}\right)^+
        \le 0
    \end{equation*}
    
    The above, because $(\hat R_k^T - \bar R_k^T)^+\ge 0$, $\hat R_k^{T_0} - \bar R_k^{T_0} \le 0$ for $k\notin I$ (from \eqref{eq:gen:91}), and $\sum_{i\in I}(\hat R_i^{T_0} - \bar R_i^{T_0})^+ = x$, proves that $\sum_{i\in I}(\hat R_i^T - \bar R_i^T) \le x$. This completes the proof.
\end{proof}

Now \cref{thm:gen:no_over_report} can be viewed as a corollary of the above lemma.
Finally, we prove the inventive compatibility bound for adversarial demands.

\generalupper*

\begin{proof}
    Because of \Cref{thm:gen:no_over_report} we assume without loss of generality that users do not over-report their demand. This entails that we can focus on the users' allocations instead of their utilities. 
    
    Fix an epoch $t$ and let $T\le t$ be the last time where $\sum_{i\in I}\hat r_i^t > \sum_{i\in I}r_i^t$. We notice that
    \begin{equation}\label{eq:gen:60}
        \sum_{i\in I}\left(\hat R_i^t - R_i^t\right)
        \;=\;
        \sum_{i\in I}\left(\hat R_i^T - R_i^T\right) + \sum_{\tau=T+1}^t\sum_{i\in I}\left(\hat r_i^T - r_i^T\right)
        \;\le\;
        \sum_{i\in I}\left(\hat R_i^T - R_i^T\right)
    \end{equation}
    
    Because $\sum_{k\in I}\hat r_k^T > \sum_{k\in I}r_k^T$ and $\sum_{k\in [n]}\hat r_k^T \le \sum_{k\in [n]}r_k^T$ (since users do not over-report their demand) there must exist a $i\in I$ such that $\hat r_i^T > r_i^T$ and a $j\notin I$ such that $\hat r_j^T < r_j^T$. Because $\hat d_i^T\le d_i^T$ and $\hat d_j^T = d_j^T$ we can use \Cref{lem:gen:aux_bound} to get
    \begin{equation}\label{eq:gen:61}
        \frac{w_j}{w_i}\left(\hat R_i^T - R_i^T\right) \le \hat R_j^T - R_j^T
    \end{equation}
    
    If we use \Cref{lem:gen:aux_bound} and sum \eqref{eq:gen:aux} for epochs up to $T$ we get
    \begin{equation*}
        \sum_{k\in [n]}(\hat R_k^T - R_k^T)^+ \le \sum_{k\in I} R_k^T
    \end{equation*}
    which combined with \eqref{eq:gen:61} gives
    \begin{equation}\label{eq:gen:62}
        \frac{w_j}{w_i}\left(\hat R_i^T - R_i^T\right)^+ + \sum_{k\in I}(\hat R_k^T - R_k^T) \le \sum_{k\in I} R_k^T \ou[\le]{t\ge T}{} \sum_{k\in I} R_k^t
    \end{equation}
    
    \eqref{eq:gen:60} and \eqref{eq:gen:62} prove that $\sum_{k\in I}\hat R_k^t \le 2\sum_{k\in I}R_k^t$, which proves the first part of the lemma.
    
    If $|I|=1$, then $I=\{i\}$ where $i$ is the user appearing in \eqref{eq:gen:62}. If $\hat R_i^T-R_i^T < 0$ the desired bound is true; otherwise, \eqref{eq:gen:60} and \eqref{eq:gen:62} prove that
    \begin{equation*}
        \min_{j\neq i}\frac{w_i}{w_j}\left(\hat R_i^t - R_i^t\right) + \hat R_i^t - R_i^t \le R_i^t
    \end{equation*}
    which proves the desired bound: $\hat R_i^t\le \left(1 + \max_{j\neq i}\frac{w_i}{w_i + w_j}\right) R_i^t$.
\end{proof}
\section{Deferred Proof of Section \ref{sec:time}}
\label{sec:app:time}

We first prove a more general version of \Cref{lem:single:aux_bound}.

\begin{lemma}\label{lem:time:aux_bound}
    Fix an epoch $t$ and the allocations of two different outcomes $\{\hat R_k^{t-1}\}_{k\in [n]}$ and $\{\bar R_k^{t-1}\}_{k\in [n]}$. Assume that $\{\bar d_k^t\}_{k\in [n]}$ are some users' demands and that $\{\hat d_k^t\}_{k\in [n]}$ are the same demands except user $1$'s, who deviates but not by over-reporting, i.e., $\hat d_1^t \le \bar d_1^t$. Then, in \nameS for any $\a\in[0,1]$, it holds that
    \begin{equation}\label{eq:time:aux_old}
        \sum_{k\in [n]}\left(\hat R_k^t - R_k^t\right)^+
        -
        \sum_{k\in [n]}\left(\hat R_k^{t-1} - R_k^{t-1}\right)^+
        \le
        f^t
    \end{equation}
    where $f^t =\min\left( (r_1^t - \hat r_1^t)^+,\,(R_1^t - \hat R_1^t)^+ \right)$.
\end{lemma}

\begin{proof}
    Let $P^t = \{k: \hat R_k^t > R_k^t\}$ and suppose by contradiction that
    \begin{equation}\label{eq:time:18}
        \sum_{k\in P^t} \left( \hat R_k^t -  R_k^t \right)^+
        -
        \sum_{k\in P^{t-1}} \left( \hat R_k^{t-1} -  R_k^{t-1} \right)^+
        >
        f^t
    \end{equation}

    In order to get a contradiction we are going to show that the following two conditions are implied by \eqref{eq:time:18}:
    \begin{enumerate}[label=(\Roman*)]
        \item\label{cond:1} There exists a user $i\in P^t$, such that $\hat r_i^t > r_i^t$.
        \item\label{cond:2} There exists a user $j\notin P^t$, such that $j\neq 1$ and $\hat r_j^t < r_j^t$.
    \end{enumerate}
    The reason conditions \ref{cond:1} and \ref{cond:2} lead to a contradiction is the following: users $i$ and $j$ satisfy the conditions of \Cref{lem:single:more_less}, meaning that $\hat R_i^t - R_i^t \leq \hat R_j^t - R_j^t$.
    However this is a contradiction due to the facts that $i\in P^t$ and $j\notin P^t$

    To prove that \eqref{eq:time:18} implies conditions \ref{cond:1} and \ref{cond:2}, we will now distinguish two cases; these are shown in \Cref{prop:aux1,prop:aux2}.

    \begin{proposition}\label{prop:aux1}
        If $\hat r_1^t \geq r_1^t$ or $\hat R_1^t \geq R_1^t$, \eqref{eq:time:18} implies conditions \ref{cond:1} and \ref{cond:2}.
    \end{proposition}

    \begin{proof}[Proof of \Cref{prop:aux1}]
        In this case we have that $f^t = 0$. Because of the definition of $P^t$, we can re-write \eqref{eq:time:18}
        \begin{equation}\label{eq:time:19}
            \sum_{k\in P^t} \hat r_k^t - r_k^t
            > 0
        \end{equation}

        \eqref{eq:time:19} implies \ref{cond:1}. \eqref{eq:time:19} and the fact that $\sum_k r_k^t \geq \sum_k \hat r_k^t$ (which comes from user $1$ not over-reporting her demand) prove that
        \begin{equation*}
            \sum_{k\notin P^t} r_k^t - \hat r_k^t
            > 0
        \end{equation*}

        The above implies condition \ref{cond:2}: there exists a $j\notin P^t$ such that $\hat r_j^t < r_j^t$. The reason that the aforementioned $j$ cannot be user $1$ is because of our assumptions: either $\hat r_1^t \geq r_1^t$ or $\hat R_1^t \geq R_1^t$.
    \end{proof}

    \begin{proposition}\label{prop:aux2}
        If $\hat r_1^t < r_1^t$ and $\hat R_1^t < R_1^t$, \eqref{eq:time:18} implies conditions \ref{cond:1} and \ref{cond:2}.
    \end{proposition}

    \begin{proof}[Proof of \Cref{prop:aux2}]
        Because $\hat R_1^t < R_1^t$, it holds that $1\notin P^t$. Consider two cases:
        \begin{itemize}
            \item If $\hat R_1^{t-1} < R_1^{t-1}$, then $f^t = r_1^t - \hat r_1^t$. In this case \eqref{eq:time:18} implies
            \begin{equation*}
                \sum_{k\in P^t} \hat r_k^t - r_k^t > r_1^t - \hat r_1^t > 0
            \end{equation*}

            \item If $\hat R_1^{t-1} \geq R_1^{t-1}$, then $f^t = R_1^t - \hat R_1^t$ and $1\in P^{t-1}$. In this case \eqref{eq:time:18} implies
            \begin{eqnarray*}
                \sum_{k\in P^t} \left( \hat R_k^t -  R_k^t \right)
                -
                \sum_{k\in P^{t-1}} \left( \hat R_k^{t-1} -  R_k^{t-1} \right)
                &>&
                R_1^t - \hat R_1^t\\
                \sum_{k\in P^t} \left( \hat R_k^t -  R_k^t \right)
                -
                \sum_{k\in P^{t-1}\setminus\{1\}} \left( \hat R_k^{t-1} -  R_k^{t-1} \right)
                &>&
                r_1^t - \hat r_1^t\\
                \sum_{k\in P^t} \hat r_k^t - r_k^t &>& r_1^t - \hat r_1^t > 0
            \end{eqnarray*}
            where to get the last inequality we use the fact that $1\notin P^t$.
        \end{itemize}

        Thus in both cases the following inequality holds:
        \begin{equation}\label{eq:time:20}
            \sum_{k\in P^t} \hat r_k^t - r_k^t > r_1^t - \hat r_1^t > 0
        \end{equation}

        \eqref{eq:time:20} implies \ref{cond:1}. \eqref{eq:time:20} and $\sum_k r_k^t \geq \sum_k \hat r_k^t$ (which is true because user $1$ does not over-report) prove that
        \begin{eqnarray*}
            \sum_{k\notin P^t} r_k^t - \hat r_k^t
            &>& r_1^t - \hat r_1^t\\
            \sum_{\substack{k\notin P^t\\k\neq 1}} r_k^t - \hat r_k^t
            &>& 0
        \end{eqnarray*}

        The above implies condition \ref{cond:2}, which completes the proposition's proof.
    \end{proof}

    Due to \Cref{prop:aux1,prop:aux2} we have proven that \eqref{eq:time:18} always leads to a contradiction. This proves the lemma.
\end{proof}

We now use the above lemma to prove a corollary that directly bounds $\hat R_1^t - R_1^t$.

\begin{corollary}\label{cor:time:aux_bound}
    Let $f^t =\min\left( (r_1^t - \hat r_1^t)^+,\,(R_1^t - \hat R_1^t)^+ \right)$ and assume that user $1$ does not over-report her demand. Then, in \nameS for any $\a\in[0,1]$ and every epoch $t$,
    \begin{equation}\label{eq:time:aux}
        2\left(\hat R_1^t - R_1^t\right)
        \le
        \sum_{\tau=1}^t f^\tau
    \end{equation}
\end{corollary}

\begin{proof}
    Fix an epoch $t$ and let $t'\le t$ be the last epoch before $t$ where $\hat r_1^{t'} > r_1^{t'}$. If no such epoch exists, then $\hat R_1^t \le R_1^t$, in which case the lemma holds. We notice that $\hat R_1^{t} - R_1^{t} \le \hat R_1^{t'} - R_1^{t'}$.
    
    Using \Cref{lem:time:aux_bound} and summing \eqref{eq:time:aux_old} for $t$ from $1$ to $t'$ we get
    \begin{equation*}
        \sum_k\left(\hat R_k^{t'} - R_k^{t'}\right)^+
        \le
        \sum_{\tau=1}^{t'}f^\tau
    \end{equation*}
    
    Because $\hat r_1^{t'} > r_1^{t'}$ and $\sum_i r_i^{t'} \ge \sum_i \hat r_i^{t'}$ (since user $1$ does not over-report) it holds that for some user $j\neq 1$, $\hat r_j^{t'} < r_j^{t'}$. This means we can use \Cref{lem:single:aux_bound} to prove that $\hat R_1^{t'} - R_1^{t'} \le \hat R_j^{t'} - R_j^{t'}$, which makes the above inequality
    \begin{equation*}
        2(\hat R_1^t - R_1^t)
        \le
        2(\hat R_1^{t'} - R_1^{t'})
        \le
        \hat R_1^{t'} - R_1^{t'} + \hat R_j^{t'} - R_j^{t'}
        \le
        \sum_{\tau=1}^{t'}f^\tau
        \le
        \sum_{\tau=1}^t f^\tau
    \end{equation*}
    
    This completes the proof.
\end{proof}

Now we prove a series of lemmas with the notation introduced in \Cref{sec:time}, in order to prove \Cref{thm:time:repeat}, which we restate for completeness.

\repeat*

\begin{lemma}\label{lem:time:prev_batch}
    If user $1$ does not over-report her demand, for any $\ell=0,1,2,\dots$ and any $t_\ell\in [s_\ell, e_\ell)$ it holds that
    \begin{equation*}
        2\left(\hat R_1^{t_\ell} - R_1^{t_\ell}\right)
        \leq
        R_1^{s_\ell-1} - \sum_{k=0}^{\ell-1}\left(\hat R_1^{e_k} - R_1^{s_k-1}\right)
    \end{equation*}
\end{lemma}

\begin{proof}
    We are going to use \Cref{cor:time:aux_bound}: for every $k\in[0,\ell-1]$ and $t\in[s_k, e_k)$ it holds that the r.h.s. of \eqref{eq:time:aux} is $f^t=0$, because $\hat R_1^t \geq R_1^t$. Fix a $t_\ell\in[s_\ell,e_\ell)$ and we notice that
    \begin{eqnarray*}
        2\left(\hat R_1^{t_\ell} - R_1^{t_\ell}\right)
        \le
        \sum_{\tau=1}^{t_\ell} f^\tau
        &=&
            \sum_{\tau=1}^{s_0-1} f^\tau +
            \sum_{\tau=e_0}^{s_1-1} f^\tau + \ldots +
            \sum_{\tau=e_{\ell-1}}^{s_\ell-1} f^\tau \\
        &\ou[\le]{f^\tau\le r_1^\tau}{}& 
            R_1^{s_0-1} + 
            \left( R_1^{s_1-1} - R_1^{e_0} + f^{e_0} \right) + \ldots +
            \left( R_1^{s_\ell-1} - R_1^{e_{\ell-1}} + f^{e_{\ell-1}} \right) \\
        &=&
            R_1^{s_\ell-1} - \sum_{k=0}^{\ell-1} \left( R_1^{e_k} - g_{e_k} - R_1^{s_k-1} \right)
    \end{eqnarray*}

    Now all that is left to complete the proof is to prove that for every $k$, $R_1^{e_k} - f^{e_k} \geq \hat R_1^{e_k}$. This actually holds with an equality: due to the definition of $e_k$, $\hat R_1^{e_k-1} > R_1^{e_k-1}$ and $\hat R_1^{e_k} \leq R_1^{e_k}$, meaning that $f^{e_k} = R_1^{e_k} - \hat R_1^{e_k}$.
\end{proof}

\begin{lemma}\label{lem:time:rec_ineq}
    Assume that for every $\ell=0,1,\ldots$ there exists a $t_\ell\in [s_\ell,e_\ell)$ such that $\hat R_1^{t_\ell} \geq \g R_1^{t_\ell}$ for some $\g\in[1, 3/2)$. Then for any such $\{t_\ell\}_\ell$ and any $\ell\geq 1$:
    \begin{equation*}
        R_1^{t_\ell}
        \geq
        \frac{\g-1}{3-2\g}\sum_{k=0}^{\ell-1} R_1^{t_{k}}
    \end{equation*}
\end{lemma}

\begin{proof}
    Fix an $\ell\geq 1$. We use \Cref{lem:time:prev_batch} and get that
    \begin{equation*}
        2\left(\hat R_1^{t_\ell} - R_1^{t_\ell}\right)
        \leq
        R_1^{s_\ell-1} - \sum_{k=0}^{\ell-1}\left(\hat R_1^{e_k} - R_1^{s_k-1}\right)
        \leq
        R_1^{t_\ell} - \sum_{k=0}^{\ell-1}\left(\hat R_1^{t_k} - R_1^{t_k}\right)
    \end{equation*}

    Using the fact that for every $k=0,\ldots,\ell$, $\hat R_1^{t_k} \geq \g R_1^{t_k}$, the above inequality becomes
    \begin{equation*}
        (2\g-3) R_1^{t_\ell}
        \leq
        - (\g-1)\sum_{k=0}^{\ell-1} R_1^{t_\ell}
    \end{equation*}
    which proves the lemma.
\end{proof}

\begin{corollary}
    If the conditions of \Cref{lem:time:rec_ineq} hold, then for all $\ell\geq 1$,
    \begin{equation}\label{eq:time:rec_sol}
        R_1^{t_\ell}
        \geq
        \frac{\g-1}{2-\g}\left(\frac{2-\g}{3-2\g}\right)^\ell
        R_1^{t_0}
    \end{equation}
\end{corollary}

\begin{proof}
    We will prove the corollary with induction on $\ell$. For $\ell=1$, \eqref{eq:time:rec_sol} follows from \Cref{lem:time:rec_ineq}.

    Assume that for some $L$, \eqref{eq:time:rec_sol} holds for all $\ell=1,2,\ldots,L$. Using \Cref{lem:time:rec_ineq} we have that
    \begin{equation*}
        R_1^{t_{L+1}}
        \geq
        \frac{\g-1}{3-2\g}\sum_{\ell=0}^L R_1^{t_\ell}
        \ou[\geq]{\eqref{eq:time:rec_sol}}{}
        R_1^{t_0}
        \frac{\g-1}{3-2\g}
        \left( 1 + \sum_{\ell=1}^{L} \frac{\g-1}{2-\g}\left(\frac{2-\g}{3-2\g}\right)^\ell \right)
        = R_1^{t_0} \frac{\g-1}{2-\g}\left(\frac{2-\g}{3-2\g}\right)^{L+1}
    \end{equation*}
\end{proof}

\begin{corollary}
    If the conditions of \Cref{lem:time:rec_ineq} hold and for all $t$, $R_1^t \in \Theta(t)$, then for all $\ell\ge 1$
    \begin{equation*}
        t_\ell = \left(\frac{2-\g}{3-2\g}\right)^\ell \Omega(t_0)
    \end{equation*}
\end{corollary}

This corollary proves the second part of \Cref{thm:time:repeat}. Now we are going to prove the first part.

\begin{lemma}
    Fix an $\ell\in\{0,1,\ldots\}$ and assume that for some $t_\ell\in [s_\ell,e_\ell)$ it holds that $\hat R_1^{t_\ell} \geq \g R_l^{t_\ell}$, for some constant $\g>1$. If $R_1^t\in\Theta(t)$ for all $t$, then for all $\ell\in\{0,1,\ldots\}$
    \begin{equation*}
        t_\ell \in O(s_\ell)
    \end{equation*}
\end{lemma}

\begin{proof}
    Using \Cref{lem:time:prev_batch} and $\hat R_1^{t_\ell} \geq \g R_1^{t_\ell}$ we can easily prove that
    \begin{equation*}
        2(\g-1) R_1^{t_\ell} \leq R_1^{s_\ell-1}
    \end{equation*}

    The lemma follows from the facts that $R_1^{t_\ell} = \Theta(t_\ell)$ and $R_1^{s_\ell-1} = \Theta(s_\ell)$.
\end{proof}




\end{document}